\documentclass[aps,pra,twocolumn,nofootinbib,10pt, superscriptaddress]{revtex4-1}
\usepackage{graphicx}
\usepackage{epstopdf}
\usepackage{amsmath}
\usepackage{amssymb}
\usepackage{mathrsfs}
\usepackage{amsthm}

\usepackage{chngcntr}
\usepackage{bm}
\usepackage{url}
\usepackage[T1]{fontenc}
\usepackage{csquotes}
\MakeOuterQuote{"}
\newtheoremstyle{note}
  {\topsep/2}               % ABOVE SPACE
  {\topsep/2}               % BELOW SPACE
  {}                      % BODY FONT
  {\parindent}            % INDENT (empty value is the same as 0pt)
  {\itshape}              % HEAD FONT
  {.}                     % HEAD PUNCTUATION
  {5pt plus 1pt minus 1pt}% HEAD SPACE
  {}

\theoremstyle{note}
\newtheorem{theorem}{Theorem}
\newtheorem{lemma}{Lemma}

\newtheorem{proposition}{Proposition}

\theoremstyle{definition}

\theoremstyle{remark}
\newtheorem{remark}{Remark}

\newcommand{\rk}{\mathrm{rk}}
\newcommand{\scf}{\mathscr{F}_\mrm{sc}}
\newcommand{\ucf}{\mathscr{F}_\mrm{uc}}
\newcommand{\syf}{\mathscr{F}_\mrm{s}}
\newcommand{\unf}{\mathscr{F}_\mrm{u}}
\newcommand{\cnot}{\mathcal{U}_{\mrm{CNOT}}}

\newcommand{\mrm}[1]{\mathrm{#1}}

\newcommand{\tr}{\operatorname{tr}}

\newcommand{\diag}{\operatorname{diag}}

\newcommand{\eig}{\operatorname{eig}}

\newcommand{\id}{I}

\newcommand{\rmd}{\mathrm{d}}
\newcommand{\rme}{\mathrm{e}}
\newcommand{\rmi}{\mathrm{i}}

\newcommand{\rmr}{\mathrm{r}}

\newcommand{\rmA}{\mathrm{A}}
\newcommand{\rmB}{\mathrm{B}}
\newcommand{\rmC}{\mathrm{C}}

\newcommand{\rmF}{\mathrm{F}}
\newcommand{\rmG}{\mathrm{G}}

\newcommand{\rmT}{\mathrm{T}}
\newcommand{\rmU}{\mathrm{U}}

\newcommand{\gc}{\mathrm{gc}}

\newcommand{\bbC}{\mathbb{C}}

 \newcommand{\caI}{\mathcal{I}}
 \newcommand{\caN}{\mathcal{N}}
 \newcommand{\caR}{\mathcal{R}}
 \newcommand{\caS}{\mathcal{S}}

\newcommand{\be}{\begin{equation}}
\newcommand{\ee}{\end{equation}}
\newcommand{\ba}{\begin{align}}
\newcommand{\ea}{\end{align}}

\def\<{\langle}  %% overiding the original command \<
\def\>{\rangle}  %% overiding the original command \>

\newcommand{\ketbra}[2]{| #1\>\< #2|}

\def\outer#1#2{|#1\>\<#2|}       %% overiding the original command \outer

\newcommand{\eref}[1]{Eq.~\textup{(\ref{#1})}}
\newcommand{\Eref}[1]{Equation~\textup{(\ref{#1})}}
\newcommand{\esref}[1]{Eqs.~\textup{(\ref{#1})}}

\newcommand{\fref}[1]{Fig.~\ref{#1}}

\newcommand{\sref}[1]{Sec.~\ref{#1}}
\newcommand{\Sref}[1]{Section~\ref{#1}}

\newcommand{\thref}[1]{Theorem~\ref{#1}}
\newcommand{\Thref}[1]{Theorem~\ref{#1}}
\newcommand{\thsref}[1]{Theorems~\ref{#1}}
\newcommand{\Thsref}[1]{Theorems~\ref{#1}}

\newcommand{\lref}[1]{Lemma~\ref{#1}}
\newcommand{\Lref}[1]{Lemma~\ref{#1}}

\newcommand{\cref}[1]{Conjecture~\ref{#1}}
\newcommand{\Cref}[1]{Conjecture~\ref{#1}}

\newcommand{\pref}[1]{Proposition~\ref{#1}}
\newcommand{\Pref}[1]{Proposition~\ref{#1}}
\newcommand{\psref}[1]{Propositions~\ref{#1}}

\newcommand{\rcite}[1]{Ref.~\cite{#1}}
\newcommand{\rscite}[1]{Refs.~\cite{#1}}

\begin{document}

\title{Operational one-to-one mapping between coherence and entanglement measures}
\author{Huangjun Zhu}
\email{These authors contributed equally to this work.}
\affiliation{Institute for Theoretical Physics, University of Cologne,
Cologne 50937, Germany}
\email{hzhu1@uni-koeln.de}

\author{Zhihao Ma}
\email{These authors contributed equally to this work.}
\affiliation{Department of Mathematics, Shanghai Jiaotong University, Shanghai, 200240, China}
\email{ma9452316@gmail.com}

\author{Zhu Cao}
\affiliation{Center for Quantum Information, Institute for Interdisciplinary Information Sciences, Tsinghua University, Beijing 100084, China}

\author{Shao-Ming Fei}
\affiliation{School of Mathematical Sciences, Capital Normal University, Beijing 100048, China}

\affiliation{Max-Planck-Institute for Mathematics in the Sciences, 04103 Leipzig, Germany}

\author{Vlatko Vedral}
\affiliation{Department of Physics, University of Oxford, Parks Road, Oxford, OX1 3PU, UK}

\affiliation{Centre for Quantum Technologies, National University of Singapore, 3 Science Drive 2, Singapore 117543, Singapore}

\begin{abstract}

We establish a general operational one-to-one mapping between  coherence measures and entanglement measures: Any entanglement measure of bipartite pure states is the minimum of a suitable coherence measure over product bases. Any coherence measure
of pure states, with extension to mixed states by convex roof, is  the maximum entanglement generated by incoherent operations acting on the system and an incoherent ancilla. Remarkably, the generalized CNOT gate is the universal optimal incoherent operation. 
In this way,  all convex-roof coherence measures, including the  coherence of formation, are endowed with (additional) operational interpretations. By virtue of this connection, many results on entanglement can be translated to the coherence setting, and vice versa. As applications, we provide tight observable lower bounds for generalized entanglement concurrence and coherence concurrence, which enable experimentalists to quantify entanglement and coherence of the maximal dimension in real experiments.

\end{abstract}

\date{\today}
\maketitle

\section{Introduction}

Quantum entanglement is a crucial resource for  many quantum information processing tasks, such as quantum teleportation, dense coding, and quantum key distribution; see \rcite{Horodecki09} for a review. It is also a useful tool for studying various intriguing  phenomena in  many-body physics and high energy physics, such as quantum phase transition  and black hole information paradox.

Quantum coherence underlies entanglement and 
is even more fundamental. It plays a key role in various research areas, such as
 interference \cite{Mandel95,Aberg06,BagaBCH16,BiswGW17}, laser \cite{Mandel95}, quantum metrology \cite{Giovannetti04,EschMD11, MarvS16}, quantum computation \cite{Shor95,Hill16,MaYGV16,Matera16}, quantum thermodynamics  \cite{HoroO13F,Aberg14,LostKJR15, Cwik15, Mitchison15,BrasB15, Goold16, Llobet17}, and photosynthesis \cite{ Engel07,Cheng09}. However,
the significance of coherence as a resource was not fully appreciated  until the works of Aberg~\cite{Aberg06} and  Baumgratz et al. \cite{Baumgratz14}, which studied coherence  from the perspective of resource theories \cite{HoroO13, Bran15, CoecFS16,LiuHL17, WinterYang16, StreAP16}. Coherence has since  found increasing applications and attracted increasing attention. Accordingly, great efforts have been devoted to quantifying coherence,
and a number of  useful coherence measures have been proposed and studied \cite{Aberg06,Baumgratz14,DuBQ15,YuanZCM15,WinterYang16, Chitambar16PRA,Napoli16,Piani16,QiGY16,Chin17, MarvS16,Chitambar16PRL,ChitambarH16,Matera16,Vicente17,MaYGV16,ChengH15,Singh15,Rana16,BiswGW17,BromCA15,YadiV16,YuZXT16,StreAP16}; 
see \rcite{StreAP16} for an overview.

The resource theory of coherence is closely related to the resource theory of entanglement \cite{Aberg06, Baumgratz14, Asboth05,StreSDB15, VogeS14, DuBG15,DuBQ15,QiBD15,XiLF15, YuanZCM15, YaoXGS15,WinterYang16,Chitambar16PRL,Chitambar16PRA,Napoli16,Piani16,QiGY16,ChitambarSRB16,MaYGV16, StreltsovCRB16,Killoran16,ChitambarH16,StreAP16,AdesBC16,ChenGJL16,Chin17}. Many results on coherence theory are inspired by analogs 
on entanglement theory, including 
many coherence measures, such as the relative entropy of coherence  (equal to the distillable coherence) \cite{Aberg06,Baumgratz14,WinterYang16}, coherence of formation (equal to the coherence cost) \cite{Aberg06, YuanZCM15,WinterYang16}, and  robustness of coherence  \cite{Napoli16,Piani16}. In addition, coherence transformations under incoherent operations are surprisingly similar to entanglement transformations under local operations and classical communication (LOCC)
\cite{DuBQ15,DuBG15, YuanZCM15, WinterYang16, Chitambar16PRA,StreAP16}. Furthermore, coherence and entanglement can be converted to each other under certain scenarios of special interest \cite{Aberg06, Asboth05,StreSDB15, VogeS14,  MaYGV16, Killoran16,QiGY16,Chin17}. 
In \rscite{Asboth05,StreSDB15}, it was shown that any degree of coherence in some reference basis can be converted to entanglement via incoherent operations. In addition, this procedure can induce coherence measures, including the relative entropy of coherence  and geometric coherence, 
from   entanglement measures \cite{StreSDB15}. However, little is known about which measures can be induced in this way beyond a few examples, and the connection between coherence and entanglement is far from clear.

In this paper, we  show that any entanglement measure of  bipartite pure states is the minimum of a suitable coherence measure over product bases. Conversely, any coherence measure
of pure states, with extension to mixed states by convex roof,  is equal to the maximum entanglement generated by incoherent operations acting on the system and an incoherent ancilla. 
Remarkably, the generalized CNOT gate is the universal optimal incoherent operation, as illustrated in \fref{fig:mapping}. 
In this way we endow all convex-roof coherence measures  with  operational meanings, including the coherence of formation \cite{Aberg06,YuanZCM15,WinterYang16} and (generalized) coherence concurrence 
\cite{QiGY16,Chin17}. In addition, our work  is instrumental in  studying interconversion between coherence and entanglement.

By virtue of the connection established here, many results on entanglement  detection and quantification can be translated to the coherence setting, and vice versa, which has
wide applications in  quantum information processing. As an illustration, we provide tight observable lower bounds for the generalized entanglement concurrence~\cite{Gour05} in terms of the negativity and robustness of entanglement. In parallel, we also provide  tight observable lower bounds for the generalized 
coherence concurrence~\cite{Chin17} in terms of the $l_1$-norm coherence and robustness of coherence.  Remarkably,  these lower bounds can be estimated in a way that is device independent. These results are useful in detecting and quantifying entanglement and coherence of the maximal dimension in real experiments.

The rest of the paper is organized as follows. In \sref{sec:Pre}, we review the general frameworks for constructing   entanglement monotones (measures) and coherence monotones (measures) based on the convex roof. In \sref{sec:Mapping}, we establish an operational one-to-one mapping between coherence monotones and entanglement monotones based on the convex roof. In \sref{sec:convert}, we derive a necessary condition on converting coherence into entanglement. In \sref{sec:gconcurrence}, we derive tight observable lower bounds for the generalized entanglement concurrence and coherence concurrence. \Sref{sec:summary} summarizes this paper. The Appendices provide additional details on entanglement monotones, coherence monotones,   coherence transformations under incoherent operations (including the majorization criterion), and some technical proofs. 

 \begin{figure}
 	\includegraphics[width=7cm]{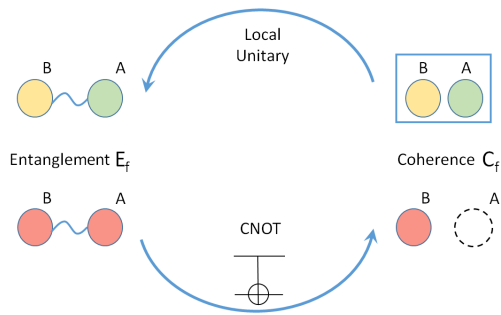}
 	\caption{\label{fig:mapping}(color online) Operational one-to-one mapping between coherence monotones and entanglement monotones. Any entanglement monotone $E_f$  for $f\in \scf$ between B and A is the minimum of $C_f$ over local unitary transformations. Any coherent monotone $C_f$ on the system B is the maximum of $E_f$ generated by incoherent operations acting on the system B and an incoherent ancilla A. The generalized CNOT gate is the universal  optimal incoherent operation. }
 \end{figure}

\section{\label{sec:Pre}Preliminaries}
\subsection{Basic concepts}
A resource theory is characterized by two basic ingredients, namely, the set of \emph{free states} and the set of \emph{free operations} \cite{HoroO13, Bran15, CoecFS16, StreAP16}. In the resource theory of entanglement,  free states are separable states,  and free operations are LOCC \cite{Horodecki09}. In the case of coherence, free states are \emph{incoherent states}, which correspond to density matrices that are diagonal in the reference basis, and free operations are \emph{incoherent operations} (IO). 
Recall that  an operation with Kraus representation $\{K_n\}$ is incoherent if each Kraus operator is incoherent in the sense that $K_n \rho K_n^\dag$ is incoherent whenever  $\rho$ is \cite{Aberg06,Baumgratz14,WinterYang16,StreAP16}. The operation is strictly incoherent  if   both $K_n$ and $K_n^\dag$ are incoherent; the set of such operations is denoted by SIO.

A central question in any resource theory is to quantify the utility  of resource states, states that are not free \cite{HoroO13, Bran15, CoecFS16, StreAP16}. Here are four typical requirements for a coherence measure $C$ \cite{Baumgratz14,StreAP16} (the situation for entanglement is analogous). (C1) Nonnegativity $C(\rho)\geq0$ (usually $C(\rho)=0$ for incoherent states); (C2) monotonicity under any incoherent operation $\Lambda$, $C(\Lambda(\rho))\leq C(\rho)$; (C3) monotonicity on average under any selective incoherent operation  $\{K_n\}$, $\sum_n p_n C(\sigma_n)\leq C(\rho)$, where $\sigma_n=K_n\rho K_n^\dag/p_n$ with $p_n=\tr(K_n\rho K_n^\dag)$; and (C4) convexity, $\sum_j q_j C(\rho_j)\geq C(\sum_j q_j\rho_j)$. Note that (C2) follows from (C3) and (C4).   A \emph{coherence monotone} satisfies   (C2-4), while a \emph{coherence measure} satisfies all (C1-4).  

\subsection{Entanglement monotones and coherence monotones based on the convex roof}

Before discussing the connection between coherence and entanglement, it is instructive to review the general  framework for constructing entanglement monotones  introduced by Vidal \cite{Vida00} and its analog for coherence \cite{DuBQ15}.
Let $\mathcal{H}$ be a $d\times d$ bipartite Hilbert space.
Denote by  $\scf$ the set of real symmetric concave functions on the probability simplex. Given any $f\in \scf$,  an entanglement monotone for  $|\psi\>\in\mathcal{H}$
can be defined as 
\begin{equation}\label{eq:EntMeasurePure}
E_f(\psi):=f(\lambda(\psi)),
\end{equation}
where $\lambda(\psi)$ is the \emph{Schmidt vector} of $\psi$, that is, the vector of Schmidt coefficients  (eigenvalues of each reduced density matrix), which form a probability vector. The  monotone  extends to mixed states by  convex roof,
\begin{equation}\label{eq:EntMeasureMix}
E_f(\rho):=\min_{\{p_j, \psi_j\}} \sum_j p_j E_f(\psi_j),
\end{equation}
where the minimum (or infimum) is taken over all pure state decompositions  $\rho=\sum_j p_j \outer{\psi_j}{\psi_j}$. The extension to systems with  different local dimensions is straightforward. The connection between
entanglement monotones and symmetric concave functions is summarized in \thref{thm:EntMeasMod} below,  which is a variant of the result presented in  \rcite{Vida00}, but tailored to highlight the connection with coherence monotones; see  Appendix~A for background and a proof.
\begin{theorem}\label{thm:EntMeasMod}
	For any $f\in \scf$, the function $E_f$ defined by \esref{eq:EntMeasurePure} and \eqref{eq:EntMeasureMix} is an entanglement monotone. Conversely, the restriction to pure states of any entanglement monotone is identical to $E_f$ for certain $f\in \scf$.
\end{theorem}

Interestingly,
  coherence monotones for pure states are also in  one-to-one correspondence with symmetric concave functions on the probability simplex \cite{DuBQ15}. Given any $f\in\scf$, a coherence monotone on $d$-dimensional  pure states
can be defined as follows,
\begin{equation}\label{eq:Cfpure}
C_f(\psi):=f(\mu(\psi)),
\end{equation}
where $\mu(\psi)=(|\psi_0|^2, |\psi_1|^2, \ldots, |\psi_{d-1}|^2)^\rmT$ is the \emph{coherence vector}, and $\psi_j$ are the components of $\psi$ in the reference  basis. For mixed states,
\begin{equation}\label{eq:Cfmix}
C_f(\rho):=\min_{\{p_j, \psi_j\}} \sum_j p_j C_f(\psi_j).
\end{equation}
This construction  is summarized in  \thref{thm:CohMeas} below, which is applicable when either IO or SIO is taken as the set of free operations. The result concerning IO
was first presented in \rcite{DuBQ15}; the original proof has a gap, but can be filled. A simple proof was given in  Appendix~B, which also leads to a simple proof of the  majorization criterion on coherent transformations \cite{DuBG15}.
\begin{theorem}\label{thm:CohMeas}
	For any $f\in\scf$, the function $C_f$ defined by \esref{eq:Cfpure} and \eqref{eq:Cfmix} is a coherence monotone. Conversely, the restriction to pure states of any coherence monotone is identical to $C_f$ for certain $f\in\scf$.
\end{theorem}

\Thsref{thm:EntMeasMod} and \ref{thm:CohMeas}
provide many useful entanglement and coherence measures. When $f(p)=-\sum_j p_j \log p_j$ denotes the Shannon entropy, $E_f$ is the celebrated entanglement of formation $E_\rmF$ (coinciding with the relative entropy of entanglement $E_\rmr$ for pure states) \cite{Horodecki09}, and $C_f$ is the coherence of formation $C_\rmF$ \cite{Aberg06,YuanZCM15} (equal to the coherence cost $C_\rmC$ \cite{WinterYang16}). 
When $f(p)=1-\max_j p_j$,  $E_f$ is the geometric entanglement $E_\rmG$ \cite{Horodecki09}, and $C_f$ is the geometric coherence $C_\rmG$ \cite{StreSDB15}. 
When $f(p)=d(\prod_j p_j)^{1/d}$,  $E_f$ and   $C_f$ reduce to the generalized entanglement concurrence~\cite{Gour05} and coherence concurrence \cite{Chin17}. These measures play important roles in theoretical studies and practical applications, so a number of methods have been developed to compute or approximate them \cite{ VollW01,ChenAF05,TothMG15,SentEGH16,GiraG17}.

\section{\label{sec:Mapping}Operational one-to-one mapping between coherence measures and entanglement measures}

The similarity between entanglement monotones  and coherence monotones reflected in  \thsref{thm:EntMeasMod} and \ref{thm:CohMeas} calls for a simple explanation. Here we shall reveal the  operational underpinning of this  resemblance.

Our study benefits from the theory of majorization \cite{MarsOA11book,Bhat97book}, which has found extensive applications in quantum information science \cite{Niel99,Vida99,ZyczB02,Turg07,DuBG15,DuBQ15,QiBD15,BuSW16}.   Given two $d$-dimensional real vectors $x=(x_0, x_1, \ldots, x_{d-1})^\rmT$ and $y=(y_0, y_1, \ldots, y_{d-1})^\rmT$, vector $x$ is majorized by $y$, written as $x\prec y$ or $y\succ x$, if
\begin{equation}
\sum_{j=0}^k x^\downarrow_j\leq \sum_{j=0}^k y^\downarrow_j \quad \forall k=0, 1, \dots, d-1,
\end{equation}
with equality for $k=d-1$. Here $x^\downarrow$ denotes the vector obtained  by arranging the components of $x$ in decreasing order.
In this work, we need to consider majorization relations between vectors of different dimensions. In such cases, it is understood implicitly that the vector with fewer components is padded with a number of  "0" to match the other vector. The notation $x\simeq y$ means that  
$x\prec y$ and $y\prec x$, so that  
 $x$ and $y$ have the same nonzero components  up to permutations.

\subsection{Entanglement as minimal coherence} 
 
Now we clarify the relation between coherence and entanglement for a bipartite state $|\psi\>$ in the Hilbert space $\mathcal{H}=\mathcal{H}_\rmB\otimes\mathcal{H}_\rmA$ of dimension $d_\rmB\times d_\rmA$. 
The reference basis   is the tensor product of respective reference bases. Denote by $E_\rk(\psi)$ the Schmidt rank of $\psi$ and $C_\rk(\psi)$ the coherence rank (number of nonzero components of $\mu(\psi)$). 
\begin{lemma}\label{lem:Majorization}
	 $\mu(\psi)\prec \lambda(\psi)$ and $C_\rk(\psi)\geq E_\rk(\psi)$ for any $|\psi\>\in \mathcal{H}_\rmB\otimes\mathcal{H}_\rmA$.  If  $\mu(\psi)\simeq \lambda(\psi)$, then $C_\rk(\psi)=E_\rk(\psi)$, and vice versa; both of them hold  if and only if (iff) $|\psi\>$ has the form
\begin{equation}\label{eq:MajorMaxKet}
|\psi\>=\sum_j \sqrt{\lambda_j(\psi)}\rme^{\rmi \theta_j}|\pi_1(j)\pi_2(j)\>,
\end{equation}
where $\theta_j$ are arbitrary  phases, and 
 $\pi_1,\pi_2$ are permutations of basis states of  
$\mathcal{H}_\rmB,\mathcal{H}_\rmA$, respectively.
\end{lemma}
\Lref{lem:Majorization} is proved in  Appendix~C. It  implies that
$\mu\left( (\mathcal{U}_1\otimes \mathcal{U}_2)(\psi) \right)\prec \lambda(\psi)$
for arbitrary local unitaries $U_1, U_2$, where $\mathcal{U}_1, \mathcal{U}_2$ denote the channels corresponding to  $U_1, U_2$. In addition,
\begin{equation}\label{eq:MajorMax}
\max_{U_1, U_2}\mu^\downarrow\left( (\mathcal{U}_1\otimes \mathcal{U}_2)(\psi) \right)\simeq \lambda^\downarrow(\psi).
\end{equation}
 Here  the maximization is taken with respect to the majorization order, which is well defined, as guaranteed by \Lref{lem:Majorization} and the Schmidt decomposition. In this way, \lref{lem:Majorization}
 offers an appealing interpretation of  the Schmidt vector in terms of the  coherence vector.

\begin{theorem}\label{thm:EntCoh} For any $f\in \scf$,
\begin{equation}\label{eq:EntCohMin}
	E_f(\rho)\leq \min_{U_1, U_2} C_f\bigl(U_1\otimes U_2 \rho (U_1\otimes U_2)^\dag\bigr)\quad \forall \rho;
	\end{equation}
the inequality is saturated if $\rho$ is pure.
\end{theorem}
The bound in \eref{eq:EntCohMin} is also saturated by maximally correlated states \cite{Rain99, StreSDB15,WinterYang16} according to \thref{thm:MC} below.

\begin{proof}
If  $\rho=\outer{\psi}{\psi}$ is   pure, then    $\mu(\psi)\prec \lambda(\psi)$ by \lref{lem:Majorization}, so  $E_f(\rho)=
f(\lambda(\psi))\leq f(\mu(\psi))=
C_f(\rho)$ since $f$ is concave and thus  Schur concave. This result  confirms \eref{eq:EntCohMin} for pure states since entanglement is invariant under local unitary transformations.
The inequality is saturated thanks to the Schmidt decomposition.

Now suppose $\rho$ is a mixed state with an optimal  decomposition $\rho=\sum_j p_j \rho_j$ with respect to $C_f$ (for simplicity, here we assume that the value of $C_f(\rho)$ can be attained by some decomposition of $\rho$, but this assumption is not essential to completing the following proof). Then
\begin{equation}
C_f(\rho)=\sum_j p_j C_f(\rho_j)\geq \sum_j p_j E_f(\rho_j)\geq E_f(\rho),
\end{equation}
from which \eref{eq:EntCohMin} follows.
\end{proof}

\subsection{Coherence as maximal entanglement}
In contrast with \Thref{thm:EntCoh}, in this section we show that every coherence monotone of pure states, with extension to mixed states by convex roof, is  the maximum entanglement generated by incoherent operations acting on the system and an incoherent ancilla. 

This line of research is inspired by a recent work of  Streltsov et al. \cite{StreSDB15}, according to which any coherent state on $\mathcal{H}_\rmB$ can generate entanglement under  incoherent operations acting on the system and an incoherent ancilla.  Moreover, the maximum entanglement $E$ generated with respect to any given entanglement monotone   defines a coherence monotone $C_E$ as follows,
\begin{equation}\label{eq:EntGen}
C_E(\rho):=\lim_{d_\rmA\rightarrow \infty}\left\{\sup_{\Lambda_\rmi}E\left(\Lambda_\rmi\left[\rho\otimes \outer{0}{0}\right]\right)\right\}.
\end{equation}
Here $d_\rmA$ is the dimension of the ancilla, and the supremum runs  over all incoherent operations. Interestingly, $C_E=C_\rmr,  C_\rmG$ when 
 $E=E_\rmr, E_\rmG$. However, little is known about other coherence monotones so constructed.

By  \eref{eq:EntGen}, 
we can introduce another coherence monotone for any symmetric concave function $f\in\scf$,
\begin{equation}
\tilde{C}_f(\rho):=C_{E_f}:=\lim_{d_\rmA\rightarrow \infty}\left\{\sup_{\Lambda_\rmi}E_f\left(\Lambda_\rmi\left[\rho\otimes \outer{0}{0}\right]\right)\right\}.
\end{equation}
Surprisingly, $\tilde{C}_f(\rho)$ coincides with $C_f$ for any $f\in \scf$. A key to establishing this result is  the generalized CNOT gate $\cnot$ corresponding to the unitary $U_{\mrm{CNOT}}$,
\begin{equation}
U_{\mrm{CNOT}}|jk\>=\begin{cases}
|j (j+k)\> & k<d_\rmB,\\
|j k\> &k\geq d_\rmB,
\end{cases}
\end{equation}
where the addition is modulo $d_\rmB$. This operation (defined when $d_\rmA\geq d_\rmB$)  turns  any state
$\rho=\sum_{jk}\rho_{jk}\outer{j}{k}$ on $\mathcal{H}_\rmB$ into a \emph{maximally correlated state} \cite{Rain99, StreSDB15,WinterYang16},
\begin{equation}
\rho_{\mrm{MC}}:=\cnot\left[\rho\otimes \outer{0}{0}\right]=\sum_{jk}\rho_{jk}\outer{jj}{kk}.
\end{equation}

\begin{theorem}\label{thm:MC}
$E_f(\rho_\mrm{MC})=C_f(\rho_\mrm{MC})=C_f(\rho)=\tilde{C}_f(\rho)$ for any $f\in \scf$.	
\end{theorem}

\begin{proof}
The equality $C_f(\rho_\mrm{MC})=C_f(\rho)$ is clear from the definition of $\rho_\mrm{MC}$.
The equality $E_f(\rho_\mrm{MC})=C_f(\rho_\mrm{MC})$  follows from the fact that  any $|\Psi\rangle$ in the support of $\rho_{\mrm{MC}}$ has the form $|\Psi\rangle=\sum_j c_j|jj\>$ with $\sum_j |c_j|^2=1$, so that $E_f(\Psi)=C_f(\Psi)$.
Therefore, $C_f(\rho)=E_f(\rho_{\mrm{MC}}) \leq \tilde{C}_f(\rho)$. The converse  $\tilde{C}_f(\rho)\leq C_f(\rho)$ holds because
	\begin{align}
	E_f\left(\Lambda_\rmi\left[\rho\otimes \outer{0}{0}\right]\right)&\leq C_f\left(\Lambda_\rmi\left[\rho\otimes \outer{0}{0}\right]\right)\nonumber\\
	&\leq C_f\left(\rho\otimes \outer{0}{0}\right)=C_f(\rho),
	\end{align}
where the first inequality follows from \thref{thm:EntCoh}, and the second one from the monotonicity of  $C_f$.
\end{proof}

\Thref{thm:MC}  endows every coherence monotone of pure states with an operational meaning as the maximal entanglement that can be generated between the system and  an incoherent ancilla under incoherent operations. This connection extends to all coherence monotones of  mixed states that are
based on the convex roof. Remarkably, the generalized CNOT gate is optimal
 with respect to all these monotones, which further implies that SIO and IO are equally powerful for entanglement generation. \Thsref{thm:EntCoh} and \ref{thm:MC} together establish a one-to-one mapping between  coherence monotones  and entanglement monotones based on the convex roof, as illustrated in \fref{fig:mapping}. Recently, this mapping was extended to many other monotones, which are not based on the convex roof \cite{ZhuHC17A, ZhuHC17C}.

\Thref{thm:MC} in particular applies to  measures based on R\'enyi $\alpha$-entropies $f(p)=(\log\sum_j p_j^\alpha)/(1-\alpha)$
with $0\leq \alpha\leq 1$ \cite{Vida00,DuBQ15,Chitambar16PRA}, which play a key role in catalytic entanglement and coherence transformations \cite{Turg07,BuSW16}. The limit $\alpha\rightarrow 1$ recovers the relation
$E_\rmC(\rho_{\mrm{MC}})=E_\rmF(\rho_{\mrm{MC}})=C_\rmF(\rho)=C_\rmC(\rho)$~\cite{WinterYang16}. \Thref{thm:MC} also implies 
$E_\rmG(\rho_{\mrm{MC}})=C_\rmG(\rho)$ \cite{StreSDB15}.
Moreover, the generalized CNOT gate is the universal optimal incoherence operation. This conclusion was  known for the geometric measure \cite{StreSDB15},   but our proof is simpler even in this case.

The power of \Thref{thm:MC} is not limited to  entanglement monotones based on convex roof. It
provides a nontrivial upper bound on entanglement  generation  for every entanglement monotone $E$. Note that, when restricted to pure states,  $E$  is determined by a symmetric concave function $f_E\in \scf$, which in turn defines an entanglement monotone $\hat{E}:=E_{f_E}$, usually referred to as the convex roof (or convex-roof  extension) of $E$.
For example,  $E_\rmF$ is the convex roof of $E_\rmr$. By construction  $\hat{E}(\sigma)\geq E(\sigma)$ for any bipartite state $\sigma$ (with equality for pure states), so $C_E(\rho)\leq C_{\hat{E}}(\rho)= C_{f_E}(\rho)$.
The same idea can also extend the scope of \thref{thm:EntCoh}.

 As another extension,  \thsref{thm:EntCoh} and  \ref{thm:MC} still apply if $E_f, C_f$ are replaced by $h(E_f), h(C_f)$ with  $h$  a real function that is monotonically increasing. In addition, the constructions in Eqs.~(\ref{eq:EntMeasurePure}-\ref{eq:Cfmix}) can be extended to functions $f$ that are Schur concave, but not necessarily concave; the resulting  quantifiers $E_f, C_f$ are not necessarily full monotones, but are useful in some applications \cite{ZyczB02,Turg07,BuSW16}. \Thref{thm:EntCoh}  holds as before, and so do the equalities $E_f(\rho_\mrm{MC})=C_f(\rho_\mrm{MC})=C_f(\rho)$ in \Thref{thm:MC}.

%%%%%%%%%%%%%%%%%%%%%%%%%%%%%%%%%%%%%%%%%%%%%%%%%%%%%%%%%%%%%%%%%%%%%%%%%%%%%%%%%%%%%%%%%%%%%%%%%%%%%%%%%%%%%%%%%%%%%%%%%%%%%%%%%%%%%%%%%%%%%%%%%%%%%%%%%%%%%%%%%%%%%%%%%%%%%%%%%%%%%%%%%%%%%%%%%%%%%%%%%

\section{\label{sec:convert}Converting coherence into entanglement}

As an application of the results presented in the previous section, here we derive a necessary condition on  converting  coherence into entanglement with incoherent operations, which is sufficient in a special case. We also derive an upper bound on the conversion probability when there is no deterministic transformation.
\begin{theorem}\label{thm:CohEntTran}
Suppose $\outer{\Phi}{\Phi}=\Lambda_\rmi\left[\outer{\psi}{\psi}\otimes \outer{0}{0}\right]$ with $\Lambda_\rmi$ being an incoherent operation. Then $\lambda(\Phi)\succ \mu(\psi)$.	If $|\Phi\>$ has Schmidt form in the reference basis,  that is, $|\Phi\>=\sum_j \sqrt{\lambda_j}|jj\>$ with $\lambda_j\geq 0$ and $\sum_j \lambda_j=1$, then $|\psi\>$ can transform to $|\Phi\>$ under IO or SIO iff $\lambda(\Phi)\succ \mu(\psi)$.
\end{theorem}
\begin{proof}
	By \lref{lem:Majorization}, $\lambda(\Phi)\succ \mu(\Phi)\succ \mu(\psi)$, where  the second inequality follows from the coherence analog of the majorization criterion \cite{Niel99,DuBG15}; cf.~\thref{thm:CohMajorization} in Appendix~B.
	
	When $|\Phi\>=\sum_j \sqrt{\lambda_j}|jj\>$, let $|\phi\>=\sum_j\sqrt{\lambda_j}|j\>$.
	If  $\lambda(\Phi)\succ \mu(\psi)$, then $\mu(\phi)\succ \mu(\psi)$, so $|\psi\>$ can transform to $|\phi\>$   under SIO \cite{DuBG15} (cf.~\thref{thm:CohMajorization}), which implies the theorem given that
	$|\Phi\>=U_{\mrm{CNOT}}(|\phi\>\otimes |0\>)$.
\end{proof}

\begin{theorem}\label{thm:MaxProb}
Let $P(\psi\rightarrow\Phi)$ be the maximal probability of generating $|\Phi\>$  from $|\psi\>$ by IO (or SIO) acting on the system and an incoherent ancilla.	Then
\begin{equation}\label{eq:MaxProb}
P(\psi\rightarrow\Phi)\leq \min_{m\geq 0}\frac{\sum_{j\geq m} \mu^\downarrow_j(\psi)}{\sum_{j\geq m} \lambda^\downarrow_j(\Phi)},
\end{equation}
with equality
if $|\Phi\>$ has the Schmidt form $\sum_j \sqrt{\lambda_j}|jj\>$.
\end{theorem}
\Thref{thm:MaxProb}  implies that  the Schmidt rank of $|\Phi\>$ can not exceed  the coherence rank of $|\psi\>$ even probabilistically.
\begin{proof}
	Define $f_m(p)=\sum_{j\geq m} p^\downarrow_j$ for  positive integers $m$.
	Then  $E_{f_m}$ and $C_{f_m}$  are entanglement measures and coherence measures according to \thsref{thm:EntMeasMod} and~\ref{thm:CohMeas}; cf.~\rscite{Vida00,DuBQ15}. 	Therefore,
	\begin{equation}
	P(\psi\rightarrow\Phi)\leq \frac{C_{f_m}(\psi)}{C_{f_m}(\Phi)}\leq \frac{C_{f_m}(\psi)}{E_{f_m}(\Phi)}=\frac{\sum_{j\geq m} \mu^\downarrow_j(\psi)}{\sum_{j\geq m} \lambda^\downarrow_j(\Phi)},
	\end{equation}
	which verifies \eref{eq:MaxProb} since $\sum_{j\geq 0} \mu^\downarrow_j(\psi)=\sum_{j\geq 0} \lambda^\downarrow_j(\Phi)$.
	
	When $|\Phi\>=\sum_j \sqrt{\lambda_j}|jj\>$, let $\phi=\sum_j\sqrt{\lambda_j}|j\>$. Then
	$|\Phi\>=U_{\mrm{CNOT}}(|\phi\>\otimes |0\>)$. Therefore
	\begin{equation}
	P(\psi\rightarrow\Phi)\geq P(\psi\rightarrow\phi)
	=\min_{m\geq 0} \frac{\sum_{j\geq m} \mu^\downarrow_j(\psi)}{\sum_{j\geq m} \mu^\downarrow_j(\phi)}.
	\end{equation}
	Here $P(\psi\rightarrow\phi)$ is the maximal probability of transforming $|\psi\>$ to $|\phi\>$ under IO (or SIO), 
	which was determined in \rcite{DuBQ15}; cf.~\thref{thm:MaxProInchoherent} in  Appendix~B. So   the inequality in \eref{eq:MaxProb} is saturated given that $\mu(\phi)\simeq\lambda(\Phi)$.
\end{proof}

\section{\label{sec:gconcurrence}Lower bounds on generalized coherence  concurrence and entanglement concurrence}
The connection between coherence and entanglement established in this work is useful not only to  theoretical studies of resource theories, but also to   practical applications in quantum information processing. By virtue of this connection, many results on entanglement  detection and quantification can  be translated  to the coherence setting, and vice versa.  As an illustration, we provide   tight observable lower bounds for the generalized entanglement  concurrence $E_\gc$~\cite{Gour05} and its coherence analog $C_\gc$ \cite{Chin17}, which correspond to the convex-roof measures $E_f$ and $C_f$ with $f(p)=d(\prod_j p_j)^{1/d}$.  Note that the definitions  of $E_\gc$ and $C_\gc$ depend explicitly on the dimension, unlike most other measures considered in this paper. 
The measure $E_\gc$ quantifies entanglement of the maximal dimension and may serve as a dimension witness.  The analog $C_\gc$ is equally important in the study of coherence.

Before presenting our main result in this section, we need to review a few coherence and entanglement measures. The $l_1$-norm  coherence
\begin{equation}
C_{l_1}(\rho): =\sum_{j\neq k}|\rho_{jk}|=\sum_{j, k}|\rho_{jk}|-1
\end{equation}
is the simplest and one of the most useful coherence measures \cite{Baumgratz14}. 
The robustness of coherence is an observable coherence measure defined as 
\begin{equation}
C_\caR(\rho):=\min\left\{x \Big|x\geq0, \; \exists \mbox{ a state } \sigma,\; \frac{\rho+x\sigma}{1+x}\in \caI \right\},
\end{equation}
where $\caI$ denotes the set of incoherent states. It has an operational interpretation in connection with  the task of phase discrimination \cite{Napoli16,Piani16}.
When $\rho$ is pure, it is known that $C_\caR(\rho)=C_{l_1}(\rho)$ \cite{Piani16,ZhuHC17C}.

The negativity of a bipartite state $\rho$ shared by B and A reads
\begin{equation}
\caN(\rho):=\tr|\rho^{\rmT_\rmA}|-1,
\end{equation}
where $\rmT_\rmA$ denotes the partial transpose on subsystem~A (the definition in some literature differs by a factor of~2). It  is essentially the only useful entanglement measure that is easily computable in general \cite{Horodecki09}. 
The robustness of entanglement is defined as 
\begin{equation}
E_\caR(\rho):=\min\left\{x \Big|x\geq0, \; \exists \mbox{ a state } \sigma,\; \frac{\rho+x\sigma}{1+x}\in \mathcal{S} \right\},
\end{equation}
where $\caS$ denotes the set of separable states. This measure has two variants: $\sigma$ is required to be separable in one variant, but could be arbitrary in the other variant \cite{Horodecki09}. When $\rho$ is pure, both  variants are equal to the negativity. \Thref{thm:EgC} below  is applicable to both cases. 

\subsection{Tight observable lower bounds}

\begin{theorem}\label{thm:CgC}Any state $\rho$ in dimension $d$ satisfies
	\begin{equation}\label{eq:CgC}
	C_\gc(\rho)+(d-2)\geq \hat{C}_{l_1}(\rho)\geq C_{l_1}(\rho)\geq C_\caR(\rho).
	\end{equation}
\end{theorem}
\begin{theorem}\label{thm:EgC}Any $d\times d$ bipartite state $\rho$ satisfies 
	\begin{equation}\label{eq:EgC}
	E_\gc(\rho)+(d-2)\geq \hat{\caN}(\rho)\geq \max\{\caN(\rho),E_\caR(\rho)\}.
	\end{equation}
\end{theorem}
Here $\hat{C}_{l_1}$ is the common convex-roof extension of $C_{l_1}$ and $C_\caR$, while $\hat{\caN}$ is the common convex-roof extension of $\caN$ and $E_\caR$. 
The inequality $C_{l_1}(\rho)\geq C_\caR(\rho)$ in \eref{eq:CgC} was derived in \rcite{Piani16}. To elucidate the connection between \thref{thm:CgC} and  \thref{thm:EgC}, let $\rho$ be a state in dimension $d$, and 
$\rho_{\mrm{MC}}:=\cnot\left[\rho\otimes \outer{0}{0}\right]$ be a $d\times d$ bipartite state. 
Then  $E_\gc(\rho_{\mrm{MC}})=C_\gc(\rho)$ according to   \thref{thm:MC}.  In addition, 
 $E_\caR(\rho_{\mrm{MC}})=C_\caR(\rho)$ according to \rcite{ZhuHC17C}; also, it is easy to verify that  $\caN(\rho_{\mrm{MC}})=C_{l_1}(\rho)$. So we have a perfect analogy between  \thref{thm:CgC} and  \thref{thm:EgC}.
 
\begin{remark}
In the above discussion,  $C_\gc(\rho)=C_f(\rho)$ with $f(p)=d(\prod_j p_j)^{1/d}$; however,  $C_\gc(\rho_{\mrm{MC}})$ is in general not
equal to  $C_f(\rho_{\mrm{MC}})$. \Thref{thm:MC} implies the equality  $E_\gc(\rho_{\mrm{MC}})=C_\gc(\rho)$, 
but cannot guarantee the equality  $C_\gc(\rho_{\mrm{MC}})=C_\gc(\rho)$ except when $C_\gc(\rho)=0$. This subtlety is tied to the fact that the definition of $C_\gc$ depends explicitly on the dimension. 
\end{remark}

All the inequalities  in \esref{eq:CgC} and \eqref{eq:EgC} can be saturated by certain states with high symmetry, as demonstrated in \sref{sec:GCsym} later.   \Thref{thm:EgC} was partially inspired by \rcite{SentEGH16}. Compared with the  lower bound for $E_\gc(\rho)$ derived in \rcite{SentEGH16},  our bound presented in \thref{thm:EgC} is much simpler and usually tighter. 
The significance  of  \thref{thm:EgC} is further strengthened by the fact that both $\caN(\rho)$ and  $E_\caR(\rho)$ are observable entanglement measures. For example, in certain scenarios of practical interest, such as in quantum simulators based on trapped ions or superconductors, a tight lower bound for $\caN(\rho)$ can be derived by measuring a single witness operator \cite{MartCP16}. In addition, $\caN(\rho)$ can be estimated in a device-independent way \cite{MoroBLH13}. Thanks to \thref{thm:EgC},   these methods can now be applied to  bound  $E_\gc(\rho)$ from below. Similarly, $C_\caR(\rho)$ can be estimated by measuring suitable  witness operators \cite{Napoli16,Piani16}, from which  we can derive  a lower bound for $C_\gc(\rho)$.

\begin{proof}[Proof of \thref{thm:CgC}]
	The inequality 	$\hat{C}_{l_1}(\rho)\geq C_{l_1}(\rho)$ follows from the convexity of $C_{l_1}$ and the definition of the convex roof. The inequality $C_{l_1}(\rho)\geq C_\caR(\rho)$ was derived in \rcite{Piani16}. 
	
	To prove the inequality $C_\gc(\rho)+(d-2)\geq \hat{C}_{l_1}(\rho)$, it suffices to consider the case in which $\rho$ is pure because both $C_\gc(\rho)$ and  $\hat{C}_{l_1}(\rho)$ are based on the convex roof. 
	Let $\rho=|\psi\>\<\psi|$  with $|\psi\>=\sum_j c_j |j\>$ and $\sum_j |c_j|^2=1$. Then  we have 
	\begin{align}
	&C_\gc(\rho)+(d-2)=d\biggl|\prod_j c_j \biggr|^{2/d}+(d-2)\nonumber\\
	&\geq \biggl(\sum_j  |c_j|\biggr)^2-1
	=C_{l_1}(\rho)=\hat{C}_{l_1}(\rho),
	\end{align}
	where the inequality follows from \lref{lem:ProductBound} below. 
\end{proof}

\begin{proof}[Proof of \thref{thm:EgC}]
	The inequality $\hat{\caN}(\rho)\geq \caN(\rho)$ is obvious. The inequality  $\hat{\caN}(\rho)\geq E_\caR(\rho)$ follows from the fact that the negativity and robustness of entanglement are convex and that  they coincide on pure states, so  they share the same convex roof, that is, $\hat{\caN}(\rho)=\hat{E}_\caR(\rho)$.

	To prove the inequality $E_\gc(\rho)+(d-2)\geq \hat{\caN}(\rho)$, 	 it suffices to consider the case in which $\rho$ is pure, as in the proof of \thref{thm:CgC}. Applying a local unitary transformation if necessary, we may assume that 
	$\rho$  has the form $\rho=|\Psi\>\<\Psi|$ with $|\Psi\>=\sum_j c_j|jj\>$, so that $\rho$ is  maximally correlated. 
	Let $\varrho=|\psi\>\<\psi|$ with $|\psi\>=\sum_j c_j|j\>$. 
Then  it is straightforward to verify that 
	$E_\gc(\rho)=C_{\gc}(\varrho)$  (cf.~\thref{thm:MC}) and $\caN(\rho)=C_{l_1}(\varrho)$. 
	Now the inequality  $E_\gc(\rho)+(d-2)\geq \hat{\caN}(\rho)$ follows from \thref{thm:CgC}.
\end{proof}

The following lemma was essentially   proved in the Supplemental Material of \rcite{SentEGH16}, though this result was not highlighted there. See  Appendix~D for a self-contained proof. 
\begin{lemma}\label{lem:ProductBound}
	Any sequence of $d$ complex numbers $c_0, c_1, \ldots, c_{d-1}$ satisfies
	\begin{equation}\label{eq:ProductBound}
	d\biggl|\prod_j c_j \biggr|^{2/d}\geq \biggl(\sum_j  |c_j|\biggr)^2-(d-1)\sum_j |c_j|^2. 
	\end{equation}
	When $d\geq3$, the inequality is saturated iff all $|c_j|$ are equal,  or all of them are equal except for one of them, which  equals 0.  
\end{lemma}

\subsection{\label{sec:GCsym}Generalized concurrence of states with high symmetry}
In this section we derive generalized coherence  concurrence and entanglement concurrence of certain states with high symmetry and thereby show that the lower  bounds for $C_\gc(\rho)$
and  $E_\gc(\rho)$ established in \thsref{thm:CgC} and \ref{thm:EgC} are tight.

Let $\rho$ be a convex combination of the maximally coherent state and the completely mixed state in dimension~$d$,
that is, 
\begin{equation}\label{eq:rhoGC}
\rho=p(|\psi\<\psi|)+(1-p)\frac{I}{d},  \quad |\psi\>=\frac{1}{\sqrt{d}}\sum_j |j\>, 
\end{equation}
with $0\leq p\leq 1$. Let  $F:=\<\psi|\rho|\psi\>=p+\frac{1-p}{d}$ be the fidelity between $\rho$ and $|\psi\>\<\psi|$; then $1/d\leq F\leq 1$. The following proposition is proved in Appendix~E.
\begin{proposition}\label{pro:GCsym} The state $\rho$ in \eref{eq:rhoGC} with $0\leq p\leq 1$ satisfies 
\begin{align}
\hat{C}_{l_1}(\rho)&= C_{l_1}(\rho)= C_\caR(\rho)=p(d-1)=d F-1, \label{eq:l1RoCsym}
\\ 
C_\gc(\rho)&=\max\{0, p(d-1)-(d-2)\}\nonumber\\
&=\max\{0, dF-(d-1)\}, \label{eq:GCsym}
\end{align}
\end{proposition}
\Pref{pro:GCsym} shows that all the inequalities in \eref{eq:CgC} of \thref{thm:CgC}  are saturated  by the state $\rho$ in \eref{eq:rhoGC} with $(d-2)/(d-1)\leq p\leq 1$, that is, $(d-1)/d\leq F\leq 1$.

Next, we show that all the inequalities in \thref{thm:EgC} are saturated by isotropic states with  sufficiently high purity. Let  $\rho$ be an isotropic state in dimension $d\times d$ \cite{Horodecki09}, which has the form
\begin{equation}
\rho=p(|\Psi\>\<\Psi|)+(1-p)\frac{I}{d^2},  \quad |\Psi\>=\frac{1}{\sqrt{d}}\sum_j |jj\>, 
\end{equation}
with $0\leq p\leq 1$. Let $F:=\<\Psi|\rho|\Psi\>=p+\frac{1-p}{d^2}$ be the fidelity between $\rho$ and $|\Psi\>\<\Psi|$. Then $1/d^2\leq F\leq 1$ and 
\begin{equation}\label{eq:rhoGEC}
\rho=F(|\Psi\>\<\Psi|)+(1-F)\frac{I-|\Psi\>\<\Psi|}{d^2-1}.
\end{equation}
The following proposition is an analog of 
\pref{pro:GCsym}. Here \eref{eq:GECsym}  follows from \rcite{SentEGH16}; \eref{eq:NRoEsym} should also be known before.  See  Appendix~E for a self-contained proof.
\begin{proposition}\label{pro:GECsym}
The state $\rho$ in \eref{eq:rhoGEC} with $F\geq 1/d^2$ satisfies 
\begin{align}
\hat{\caN}(\rho)&= \caN(\rho)= E_\caR(\rho)=\max\{0,dF-1\}, \label{eq:NRoEsym}
\\ E_\gc(\rho)&=\max\{0, dF-(d-1)\}. \label{eq:GECsym}
\end{align}
\end{proposition}
\Pref{pro:GECsym} shows that
all the inequalities in \eref{eq:EgC}  of \thref{thm:EgC}  are saturated  by the state $\rho$ in \eref{eq:rhoGEC} with  $(d-1)/d\leq F\leq 1$.

\section{\label{sec:summary}Summary}
In summary, we  established a general operational one-to-one mapping between coherence measures and entanglement measures. Any entanglement measure of bipartite pure states is the minimum of a suitable coherence measure over product bases; any coherence measure
of pure states, with  extension to mixed states by convex roof, is  the maximum entanglement generated by  incoherent operations acting on  the system and an incoherent ancilla. Besides its foundational significance in bridging the two resource theories, this connection has
wide applications in quantum information processing. Thanks to this connection, many results on entanglement can be generalized  to the coherence setting, and vice versa. As an illustration, we provided tight observable lower bounds for generalized entanglement concurrence and coherence concurrence, which enable experimentalists to quantify entanglement and coherence of the maximal dimension in real experiments.

\bigskip

\acknowledgments
We are grateful to one referee for constructive suggestions and to another referee for mentioning \rcite{SentEGH16}. 
ZM thanks Prof. Jingyun Fan for helpful discussion.
HZ acknowledges financial support
from the Excellence
Initiative of the German Federal and State Governments
(ZUK~81) and the DFG. ZM acknowledges
support from The National Natural Science Foundation
of China (NSFC), Grants No.~11275131 and No.~11571313. 
SMF acknowledges
support from NSFC, Grant No.~11675113.
VV thanks the Oxford Martin School at the University of Oxford, the Leverhulme Trust (UK), the John Templeton Foundation,  the EPSRC (UK) and the Ministry of Manpower (Singapore). This research is also supported by the National Research Foundation, Prime Ministers Office, Singapore, under its Competitive Research Programme (CRP Award No. NRF-CRP14-2014-02) and administered by the Centre for Quantum Technologies, National University of Singapore.

\appendix

\section{Entanglement monotones}
In this Appendix,  we provide additional details on the connection between entanglement monotones  and symmetric concave functions on the probability simplex. We then prove Theorem~1  in the main text, which   is a variant of  a result first established  by Vidal \cite{Vida00}. This result is now well known among the experts, but some subtlety discussed here may be helpful to other readers.

Denote by $\mathcal{T}(\bbC^d)$  the space of density matrices on $\bbC^d$ and $\rmU(d)$ the group of unitary operators on $\bbC^d$. 
Let $\unf$ be the set of unitarily invariant   functions on the space of density matrices.
We assume that each function  $f\in \unf$ is defined on $\mathcal{T}(\bbC^d)$ for each positive integer $d$. For given $d$, the function satisfies 
\begin{gather}\label{eq:unf}
f(U\rho U^\dag)= f(\rho)\quad \forall \rho\in \mathcal{T}(\bbC^d),\; U\in \rmU(d).
\end{gather}
So $f(\rho)$ is a function of the eigenvalues of $\rho$. We also assume implicitly that the number of "0" in the spectrum of $\rho$ does not affect the value of 
$f(\rho)$.  Let $\ucf\subset \unf$ be the set of unitarily invariant real concave  functions on the space of density matrices. For given $d$,
each  function $f\in \ucf$  satisfies \eref{eq:unf} and in addition
\begin{gather}
f(p\rho_1+(1-p)\rho_2)\geq p f(\rho_1)+(1-p) f(\rho_2)\nonumber\\
\forall \rho_1, \rho_2\in \mathcal{T}(\bbC^d),\; 0\leq p\leq 1.
\end{gather}

Let $\mathcal{H}= \bbC^d\otimes \bbC^d$ be a bipartite Hilbert space shared by B and A. For simplicity here we assume that the Hilbert spaces for the two subsystems have the same dimension~$d$, but  this is not essential.
Any function $f\in\ucf$ can be used to construct an entanglement monotone $E_f$ on $\mathcal{T}(\mathcal{H})$ as follows \cite{Vida00}.
For a pure state $|\psi\>\in \mathcal{H}$,
\begin{equation}\label{seq:EntMeasurePure}
E_f(\psi):=f\left(\tr_\rmA(\ketbra{\psi}{\psi})\right).
\end{equation}
The monotone is then extended to mixed states $\rho\in \mathcal{T}(\mathcal{H})$ by convex roof,
\begin{equation}\label{seq:EntMeasureMix}
E_f(\rho):=\min_{\{p_j,\rho_j\}}\sum_j p_j E_f(\rho_j),
\end{equation}
where the minimization runs over all pure state ensembles of $\rho$ for which $\rho=\sum_j p_j\rho_j$.

The following theorem is reproduced from \rcite{Vida00}, where the reader can find a detailed proof.
\begin{theorem}\label{sthm:EntMeas}
	For any $f\in \ucf$, the function $E_f$ defined by \esref{seq:EntMeasurePure} and \eqref{seq:EntMeasureMix} is an entanglement monotone. Conversely, the restriction to pure states of any entanglement monotone is identical to $E_f$ for certain $f\in \ucf$.
\end{theorem}

Next we clarify the relation between \thref{sthm:EntMeas} and Theorem~1 in the main text. Let $\Delta_d$ be the probability simplex of probability vectors with $d$ components.  A function on $\Delta_d$ is symmetric if it is invariant under permutations of the components of  probability vectors. 
Let  $\syf$ be the set of symmetric functions on the probability simplex. Here we assume implicitly  that each $f\in \syf$ is defined on $\Delta_d$ for each positive integer~$d$. In addition,  the value of $f(x)$ does not depend on the number of "0" in the components of  $x$; in other words, $f(x)=f(y)$ whenever $x\simeq y$ (which means $x\prec y$ and $y\prec x$), even if $x$ and $y$ have different numbers of components.

Any   symmetric function $f$ on the probability simplex  can be lifted to a unitarily invariant function on the space of density matrices,
\begin{equation}\label{eq:lift}
\check{f}(\rho):=f(\eig(\rho))\quad \forall \rho\in \mathcal{T}(\bbC^d).
\end{equation}
Conversely,  any  unitarily invariant function $f$ on the space of density matrices defines a symmetric function on the probability simplex  when restricted to diagonal density matrices,
\begin{equation}\label{eq:restriction}
\hat{f}(p):=f(\diag(p)) \quad \forall p\in \Delta_d.
\end{equation}
It is straightforward to verify that $\hat{\check{f}}=f$ for  any $f\in\syf$ and that 
$\check{\hat{f}}=f$ for any  $f\in\unf$. 
So the lifting map $f\mapsto \check{f}$ and the restriction map $f\mapsto \hat{f}$ establish a one-to-one correspondence between symmetric functions in $\syf$ and unitarily invariant  functions in $\unf$.

Recall that a real function 
$f$ on the probability simplex is Schur convex if it preserves the majorization order, that is, $f(x)\leq f(y)$ whenever $x\prec y$. By contrast, $f$ is Schur concave if it reverses the majorization order, that is, $f(x)\geq f(y)$ whenever $x\prec y$ \cite{MarsOA11book,Bhat97book}. 
Note that Schur convex functions and Schur concave functions are necessarily symmetric. 
In addition, symmetric convex (concave) functions are automatically Schur convex (concave), but not vice versa in general.
With this background, it is not difficult to show that the maps defined by  \esref{eq:lift} and \eqref{eq:restriction} preserve
(Schur) convexity and (Schur) concavity  for real  functions. Here we prove one of these properties that is most relevant to the current study; the other three properties  follow from a similar reasoning.  Recall that $\scf$ is the set of real symmetric concave functions on the probability simplex.
\begin{lemma}\label{lem:scf-ucf}
	The two maps $f\mapsto \check{f}$ and $f\mapsto\hat{f}$ set a bijection between $\scf$ and $\ucf$.	
\end{lemma}
\begin{proof}
	To prove the lemma, it suffices to show that the two maps $f\mapsto \check{f}$ and $f\mapsto\hat{f}$	preserve concavity.
	Given   $f\in\scf$,  let  $\rho_1, \rho_2\in \mathcal{T}(\bbC^d)$ be two arbitrary density matrices and $0\leq p\leq 1$. It is well-known that \cite{Bhat97book}
	\begin{equation}\label{eq:eigSum}
	\eig(p\rho_1+(1-p)\rho_2)\prec  p\eig^\downarrow(\rho_1)+(1-p)\eig^\downarrow(\rho_2), 
	\end{equation}
	where $\eig(\rho)$ denotes the vector of eigenvalues of $\rho$. 
	Consequently,
	\begin{align}
	&\check{f}\bigl(p\rho_1+(1-p)\rho_2\bigr)=f\bigl(\eig(p\rho_1+(1-p)\rho_2)\bigr)\nonumber\\
	&\geq f\bigl(p\eig^\downarrow(\rho_1)+(1-p)\eig^\downarrow(\rho_2)\bigr)\nonumber\\
	&\geq p f\bigl(\eig^\downarrow(\rho_1)\bigr)+(1-p)f\bigl(\eig^\downarrow(\rho_2)\bigr)\nonumber\\
	&=pf\bigl(\eig(\rho_1)\bigr)+(1-p)f\bigl(\eig(\rho_2)\bigr)\nonumber\\
	&=p\check{f}(\rho_1)+(1-p)\check{f}(\rho_2),
	\end{align}
	where the first inequality follows from \eref{eq:eigSum} and Schur concavity of $f$, and the second inequality from the concavity of $f$.
	Therefore $\check{f}$ is concave whenever $f$ is concave.
	
	On the other hand, if $f\in \ucf$, then $f$ is   concave in particular on diagonal density matrices, which implies that $\hat{f}$ is concave.
\end{proof}

\begin{proof}[Proof of Theorem~1]
	The theorem is an immediate consequence of \thref{sthm:EntMeas} and \lref{lem:scf-ucf}.
\end{proof}

\section{Coherence monotones and Coherence transformations}
 In this section we present a self-contained proof of
Theorem~2, which connects coherence monotones and symmetric concave functions on the probability simplex.  A variant of
this result  was first presented by Du, Bai, and Qi \cite{DuBQ15} (Theorem~1 there). The original proof has a gap in one direction (in particular the reasoning leading to Eq.~(12) there was not fully justified). Nevertheless, all essential ideas are already manifested in the proof.
It should be emphasized that the  set of coherence monotones is the same 
irrespective whether  IO or SIO is taken as the set of free operations.
Our study
also leads to a simpler proof of the majorization criterion on  coherent transformations under incoherent operations~\cite{DuBG15},  which is
the analog of Nielsen's majorization criterion on entanglement transformations under local operations and classical communication (LOCC) \cite{Niel99}. As we shall see, the two proofs share a key ingredient, which reflects the strong connection between coherence measures and coherence transformations.

Recall that any quantum operation $\Lambda$ (completely positive trace-preserving map) has a Kraus representation, that is, $\Lambda(\rho)=\sum_n K_n\rho K_n^\dag$, where the Kraus operators $K_n$ satisfy $\sum_n K_n^\dag K_n=\id$.
Denote by  $\mathcal{I}$  the set of incoherent states with respect to a given reference basis. Then the  operator $K_n$ is incoherent if
$K_n \rho K_n^\dag/p_n\in \mathcal{I}$  whenever  $\rho\in \mathcal{I}$ and $p_n=\tr(K_n\rho K_n^\dag)>0$. It is strictly incoherent if in addition $K_n^\dag$ is also incoherent.
Simple analysis shows that $K_n$ is incoherent iff its representation with respect to the reference basis has at most one nonzero entry in each column, and strictly incoherent if the same is also true for each row. The operation $\Lambda$ with Kraus representation $\{K_n\}$ is (strictly) incoherent if  each Kraus operator $K_n$ is (strictly) incoherent.
Although we are primarily concerned with incoherent operations (IO), most of our results concerning IO also apply to strictly  incoherent operations (SIO).

\subsection{Coherence monotones}
Recall that $\scf$  is the set of real symmetric concave functions on the probability simplex, and that each function $f\in\scf$ can be used to define a coherence monotone $C_f$ \cite{DuBQ15}. When $|\psi\>$ is a pure state,
\begin{equation}\label{seq:Cfpure}
C_f(\psi):=f(\mu(\psi));
\end{equation}
in general,
\begin{equation}
C_f(\rho):=\min_{\{p_j, \rho_j\}} \sum_j p_jC_f(\rho_j),
\end{equation}
where the minimization runs over all pure state ensembles of $\rho$ for which $\rho=\sum_j p_j\rho_j$.

\begin{proof}[Proof of Theorem~2]
	By the nature of the convex-roof construction, $C_f$ is automatically convex. In addition, to prove monotonicity under selective operations, it suffices to consider the scenario with a   pure initial state  $\psi$. Let $\Lambda=\{K_n\}$ be an arbitrary incoherent operation and  $|\varphi_n\rangle=K_n|\psi\>/\sqrt{p_n}$ with $p_n=\tr(K_n|\psi\rangle\langle\psi|K_n^\dag)$.  Then
	\begin{align}
	\sum_n p_n C_f(\varphi_n)&=\sum_n p_n f(\mu^\downarrow(\varphi_n))\leq f\left(\sum_n p_n \mu^\downarrow(\varphi_n)\right)\nonumber\\
	&\leq f( \mu(\psi))=C_f(\psi),
	\end{align}
	where the first inequality follows from the concavity of $f$, and the second inequality follows from \lref{lem:MajorizeOneMany} below and Schur concavity of $f$ (note that a symmetric concave function is automatically Schur concave). Therefore, $C_f$ is indeed a coherence monotone.

	Now we come to the converse, which is based on \rcite{DuBQ15}. Let $C$ be an arbitrary coherence monotone, then $C(\psi)$ is necessarily a symmetric function of $\mu(\psi)$ given that monomial unitaries (including permutations) are incoherent. Define $f$ on the probability simplex as follows, $f(x)=C(\psi(x))$ with $|\psi(x)\>=\sum_j \sqrt{x_j}|j\>$. Then $f$ is clearly symmetric. To prove concavity, let $x, y$ be two probability vectors, and $z=px+(1-p) y$ with $0\leq p\leq 1$. Let $|\psi(y)\>=\sum_j \sqrt{y_j}|j\>$ and $|\psi(z)\>=\sum_j\sqrt{z_j}|j\>$. Construct the quantum operation with the following two Kraus operators
	\begin{equation}
	\begin{aligned}
	K_1&=\sqrt{p} \diag\left(\sqrt{\frac{x_0}{z_0}},\sqrt{\frac{x_1}{z_1}},\ldots,\sqrt{\frac{x_{d-1}}{z_{d-1}}}\right),\\
	K_2&=\sqrt{1-p} \diag\left(\sqrt{\frac{y_0}{z_0}},\sqrt{\frac{y_1}{z_1}},\ldots,\sqrt{\frac{y_{d-1}}{z_{d-1}}}\right).
	\end{aligned}
	\end{equation}
	Here $x_j/z_j$ and $y_j/z_j$ for $0\leq j\leq d-1$ can be set to 1 whenever $z_j=0$, in which case  either $p(1-p)=0$ or $x_j=y_j=0$. Note that $K_1, K_2$ are strictly incoherent and satisfy
	$K_1^\dag K_1+K_2^\dag K_2=\id$. In addition,
	\begin{equation}
	K_1|\psi(z)\>=\sqrt{p}|\psi(x)\>,\quad K_2|\psi(z)\>=\sqrt{1-p}|\psi(y)\>.
	\end{equation}
	Since $C$ is a coherence monotone by assumption, we deduce that
	\begin{equation}
	C(\psi(z))\geq pC(\psi(x))+(1-p)C(\psi(y)),
	\end{equation}
	which implies that
	\begin{equation}
	f(px+(1-p)y)=f(z)\geq p f(x)+(1-p)f(y).
	\end{equation}
	Therefore, $f$ is both symmetric and concave. In addition, the coherence monotone $C$ coincides with   $C_f$ when restricted to pure states.
	
	Note that the above proof applies when either IO or SIO is taken as the set of free operations. Therefore, the set of convex-roof coherence monotones (measures) does not change under the interchange of IO and SIO.	
\end{proof}

\subsection{Coherent transformations under incoherent operations}
\Lref{lem:MajorizeOneMany} below was inspired by \rscite{DuBG15,DuBQ15}. It is a key ingredient for proving Theorem~2 and for establishing the majorization criterion on coherence transformations. Upon completion of this paper, we discovered that \lref{lem:MajorizeOneMany} follows from Theorem~1 in  \rcite{QiBD15}. However, the proof there crucially depends on Theorem~1 in \rcite{DuBQ15}  by the same authors, which is a variant of Theorem~2 in our main text that we try to prove. To avoid circular argument and to make our presentation self-contained, the discussion here is instrumental.

\begin{lemma}\label{lem:MajorizeOneMany}
	Suppose $|\psi\>$ is an arbitrary pure  state and $\Lambda=\{K_n\}$  is an arbitrary incoherent operation acting on $|\psi\>$.
	Let $p_n=\tr(K_n|\psi\rangle\langle\psi|K_n^\dag)$ and  $|\varphi_n\rangle=K_n|\psi\>/\sqrt{p_n}$  when $p_n>0$. Then
	\begin{equation}\label{eq:MajorizeOneMany}
	\mu(\psi)\prec \sum_n p_n\mu^\downarrow(\varphi_n).
	\end{equation}		
\end{lemma}
Although $|\varphi_n\rangle$ is not well defined when $p_n=0$, this fact does not cause any difficulty because  $K_n|\psi\>$ is what really matters in our calculation and it vanishes when $p_n=0$. Alternatively, we may restrict 
the summation in \eref{eq:MajorizeOneMany}
to the terms with $p_n>0$, and the conclusion is the same.
Similar comments also apply to several other equations appearing in this paper, but will not be mentioned again to avoid verbosity.

\begin{proof}
	By assumption each Kraus operator $K_n$ is incoherent and thus has at most one nonzero entry in each column. Therefore,  $K_n$
	can be expressed in the form $K_n=P_n \tilde{K}_n$, where $P_n$ is a permutation matrix and $\tilde{K}_n$ is upper triangular.
	The normalization condition $\sum_n K_n^\dag K_n=\sum_n\tilde{K}_n^\dag \tilde{K}_n=\id$ implies that $\sum_{j,n} {(\tilde{K}_n^*)}_{jk} {(\tilde{K}_n)}_{jl}=\delta_{kl}$ for all $k,l$.
	Since $\tilde{K}_n$ are upper triangular, we deduce that	
	\begin{equation}\label{eq:SummationFormula}
	\sum_{n}  \sum_{j=0}^{r} {(\tilde{K}_n^*)}_{jk}{(\tilde{K}_n)}_{jl}=\delta_{kl}\quad \forall r\geq \min\{k,l\}.
	\end{equation}
	Let  $|\tilde{\varphi}_n\rangle=\tilde{K}_n|\psi\>/\sqrt{p_n}$, then
	$\sqrt{p_n}{(\tilde{\varphi}_n)}_j=\sum_k {(\tilde{K}_n)}_{jk}\psi_k$, so that
	\begin{align}\label{eq:SummationFormula2}
	&\sum_{j=0}^r\sum_n  p_n\mu_j(\tilde{\varphi}_{n})
	=\sum_{j=0}^r\sum_n  p_n\bigl|{(\tilde{\varphi}_n)}_j\bigr|^2\nonumber\\
	&=\sum_{k,l}\sum_n \sum_{j=0}^r {(\tilde{K}_n^*)}_{jk}{(\tilde{K}_n)}_{jl}\psi_k^*\psi_l\nonumber\\
	&=\sum_{k=0}^r|\psi_k|^2+
	\sum_{k,l>r}\sum_n \sum_{j=0}^r {(\tilde{K}_n^*)}_{jk}{(\tilde{K}_n)}_{jl}\psi_k^*\psi_l\nonumber\\
	&=\sum_{k=0}^r|\psi_k|^2+
	\sum_n \sum_{j=0}^r \left|\sum_{l>r}{(\tilde{K}_n)}_{jl}\psi_l\right|^2\nonumber\\
	&\geq \sum_{k=0}^r|\psi_k|^2=\sum_{k=0}^r\mu_k(\psi),
	\end{align}
	where the third equality follows from \eref{eq:SummationFormula}.

	Since permutations of basis states are incoherent,
	without loss of generality, we may assume that the coefficients $|\psi_j|$ of $\psi$ in the reference basis are in decreasing order. Then  \eref{eq:SummationFormula2}  implies that
	\begin{equation}
	\mu(\psi)\prec \sum_n p_n\mu(\tilde{\varphi}_{n})\prec \sum_n p_n\mu^\downarrow(\tilde{\varphi}_{n})=\sum_n p_n\mu^\downarrow(\varphi_{n});
	\end{equation}
	here the last step  follows from the relation $|\varphi_n\>=P_n|\tilde{\varphi}_{n}\>$ with $P_n$ being a permutation, that is, $\mu(\varphi_{n})\simeq\mu(\tilde{\varphi}_{n})$.	
\end{proof}
As a side remark, \eref{eq:SummationFormula2} in the  proof of \lref{lem:MajorizeOneMany} actually holds for a larger class of operations whose Kraus operators have upper triangular form up to permutations on the left.  Such operations may generate coherence, but \eref{eq:MajorizeOneMany} still holds nevertheless if the coefficients $|\psi_j|$ of $\psi$ in the reference basis are in decreasing order. However, in general this conclusion no longer holds if  the coefficients do not have this property. Although this property can be recovered by a suitable permutation, the permutation required may destroy the upper triangular structure of the Kraus operators, which cannot be recovered by  permutations only on the left, in contrast with the scenario of incoherent operations. That is why \lref{lem:MajorizeOneMany} cannot hold in general for this wider class of operations, as expected.

\subsection{The majorization criterion on coherence transformations}
In addition to proving Theorem~2, \lref{lem:MajorizeOneMany} enables us     to construct a simple proof of
the  majorization criterion on coherence transformations under incoherent operations \cite{DuBG15}.
The result is the analog of Nielsen's majorization criterion on entanglement transformations under LOCC \cite{Niel99}.
\begin{theorem}\label{thm:CohMajorization}
	The pure state $|\psi\>$ can be transformed to  $|\varphi\>$ under IO or SIO
	iff $\mu(\psi)$ is majorized by $\mu(\varphi)$.
\end{theorem}
The conclusion concerning IO was first presented in \rcite{DuBG15};  the original proof of the "only if" part has a gap, which was corrected upon completion of our work.  In view of this gap, several recent works have derived weaker forms of \thref{thm:CohMajorization}. In particular, the conclusion concerning SIO was established in 
\rscite{WinterYang16,Chitambar16PRA}.

\begin{proof}
	Suppose $|\psi\>$ can be transformed to $|\varphi\>$ under an incoherent operation $\Lambda=\{K_n\}$. Let  $|\varphi_n\rangle=K_n|\psi\>/\sqrt{p_n}$ with $p_n=\tr(K_n|\psi\rangle\langle\psi|K_n^\dag)$. Then all $|\varphi_n\>$ with $p_n>0$ are identical to $|\varphi\>$ up to phase factors.
	So $\mu(\psi)\prec \mu(\varphi)$ according to  \lref{lem:MajorizeOneMany}, that is, $\mu(\psi)$ is majorized by $\mu(\varphi)$. Obviously, the same reasoning applies if $\Lambda=\{K_n\}$ is strictly incoherent.

	The proof of the other direction follows the approach presented in \rcite{DuBG15}. Since diagonal unitaries are incoherent, without loss of generality, we may assume that the coefficients $\psi_j, \varphi_j$ of $|\psi\>,\, |\varphi\>$ in the reference basis are real and nonnegative.
	
	If $\mu(\psi)$ is majorized by $\mu(\varphi)$, then $\mu(\psi)=A\mu(\varphi)$ with $A$  a suitable doubly stochastic matrix \cite{MarsOA11book,Bhat97book,Niel99}. Such a matrix can always be written as the product of a finite number of $T$-matrices, that is, $A=T_1 T_2\cdots T_k$, where each $T_j$ for $1\leq j\leq k$ acts nontrivially only on two components, on which it takes on the form
	\begin{equation}
	T=\begin{pmatrix}
	a & 1-a\\
	1-a &a
	\end{pmatrix},\quad 0\leq a\leq 1.
	\end{equation}
	By induction and the assumption that  permutations are free, we may assume that $A$ is a $T$-matrix of the form $A=\diag(T,\id)$ with $0<a<1$, so that $\mu(\psi)=\diag(T,\id)\mu(\varphi)$. In addition, we may assume that the first two components $\varphi_0, \varphi_1$ of $|\varphi\>$ are not zero simultaneously since, otherwise, the action would be trivial. Let
	\begin{equation}
	\begin{aligned}
	K_1&=\sqrt{a}\diag\left(\frac{\varphi_0}{\psi_0}, \frac{\varphi_1}{\psi_1}, 1,\ldots, 1\right),\\
	K_2&=\sqrt{1-a}\diag(K_2',1,\ldots,1),\;
	K_2'=\begin{pmatrix}
	0 & \frac{\varphi_0}{\psi_1}\\
	\frac{\varphi_1}{\psi_0} &0
	\end{pmatrix}.
	\end{aligned}
	\end{equation}
	Then the two operators   $K_1, K_2$ are strictly incoherent and satisfy $K_1^\dag K_1+K_2^\dag K_2=\id$.
	In addition,
	\begin{equation}
	K_1|\psi\>=\sqrt{a}|\varphi\>,\quad K_2|\psi\>=\sqrt{1-a}|\varphi\>.
	\end{equation}
	So the two operators $K_1, K_2$ define a strictly incoherent quantum operation that achieves the desired transformation from $|\psi\>$ to $|\varphi\>$.
\end{proof}

When there is no deterministic transformation from $|\psi\>$ to $|\varphi\>$,  it is of interest to determine the maximal probability of such  transformations. This problem has been solved in \rcite{DuBQ15} recently. The result is reproduced below for  the convenience of the reader.
\begin{theorem}\label{thm:MaxProInchoherent}
	Let $P(\psi\rightarrow\varphi)$ be the maximal probability of transforming  $|\psi\>$ to $|\varphi\>$ under IO. Then 
	\begin{equation}
	P(\psi\rightarrow\varphi)= \min_{m\geq 0}\frac{\sum_{j\geq m} \mu^\downarrow_j(\psi)}{\sum_{j\geq m} \mu^\downarrow_j(\varphi)}.
	\end{equation}
\end{theorem}
\Thref{thm:MaxProInchoherent} still holds if IO is replaced by SIO. This is clear from the proof presented in \rcite{DuBQ15} and \thref{thm:CohMajorization}, which
imply that the incoherent operation achieving the maximal probability can be chosen to be strictly incoherent. It is worth pointing out that the sum
$\sum_{j\geq m} \mu^\downarrow_j(\psi)$ for each positive integer  $m$ is a bona fide coherence measure associated with the function $f_m(p)=\sum_{j\geq m} p^\downarrow_j$, which is symmetric and concave \cite{Vida99,DuBQ15}.
\Thref{thm:MaxProInchoherent} is the analog of a similar result on entanglement transformations under LOCC, first established in \rcite{Vida99}. The idea of the proof in \rcite{DuBQ15} also mirrors the analog  in the entanglement setting. Not surprisingly, the majorization criterion    plays a crucial role in proving \thref{thm:MaxProInchoherent} as it does in proving the result in \rcite{Vida99}.  Also, the proof relies on  Theorem~1 in  \rcite{DuBQ15}, which is a variant of Theorem~2 in our main text. 
Since these stepping stones have been corroborated, \thref{thm:MaxProInchoherent} is  well established by now.

\section{Proof of \lref{lem:Majorization}}
\begin{proof}
	Expand $|\psi\rangle$ in the reference basis $|\psi\>=\sum_{jk}c_{jk}|jk\>$,
	then $\mu(\psi)=( |c_{jk}|^2)_{jk}$. Let $\rho_\rmB$ be the reduced density matrix for subsystem B, then   the diagonal of $\rho_\rmB$ reads $\diag(\rho_\rmB)=(\sum_k |c_{jk}|^2)_j$. It follows that  $\mu(\psi)\prec\diag(\rho_\rmB)\prec \eig(\rho_\rmB)\simeq\lambda(\psi)$, where $\eig(\rho_\rmB)$ denotes the vector of eigenvalues of $\rho_\rmB$, and we have applied the  majorization relation
	$\diag(\rho_\rmB)\prec \eig(\rho_\rmB)$~\cite{Bhat97book}. The inequality  $C_\rk(\psi)\geq E_\rk(\psi)$ is an immediate consequence of the relation $\mu(\psi)\prec\lambda(\psi)$.

	Let  $r=E_\rk(\psi)$ be the Schmidt rank  of $|\psi\>$.  If $\mu(\psi)\simeq \lambda(\psi)$, then   $C_\rk(\psi)= E_\rk(\psi)=r$.
	If $C_\rk(\psi)= E_\rk(\psi)=r$, then $|\{j |  c_{jk}\neq 0\; \exists k\}|\geq r$	 given that $\rho_\rmB$ has rank~$r$. By the same token $|\{k |  c_{jk}\neq 0\; \exists j\}|\geq r$. So the coefficient matrix $c_{jk}$
	has exactly $r$ nonzero components, with at most one  on each row and each column. Therefore, $|\psi\>$ has the form $|\psi\>=\sum_{j=0}^{r-1} a_j |\pi_1(j)\pi_2(j)\>
	$ with $\sum_j |a_j|^2=1$, where $\pi_1, \pi_2$ are two permutations of basis states. In addition,  $|a_j|^2$ coincide with the Schmidt coefficients of $|\psi\>$, which implies \eref{eq:MajorMaxKet} after redefining $\pi_1,\pi_2$ if necessary. Conversely, the relation $\mu(\psi)\simeq \lambda(\psi)$ 	 holds automatically whenever $|\psi\>$ has the form of \eref{eq:MajorMaxKet}.
\end{proof}

%%%%%%%%%%%%%%%%%%%%%%%%%%%%%%%%%%%%%%%%%%%%%%%%%%%%%%%%%%%%%%%%%%%%%%%%%%%%%%%%%%%%%%%%%%%%%%%%%%%%%%%%%%%%%%%%%%%%%%%%%%%%%%%%%%%%%%%%%%%%%%%%%%%%%%%%%%%%%%%%%%%%%%%%%%
\section{Proof of \lref{lem:ProductBound}}
\begin{proof}
\Lref{lem:ProductBound} is trivial when $d=1,2$, so we assume $d\geq3$ in the following discussion. 
Note that both sides of  the inequality in \eref{eq:ProductBound} are invariant under permutations of $c_j$ and are independent of the phase factors, so we may assume that $c_0\geq c_1\geq \cdots\geq c_{d-1}\geq0$ without loss of generality; then it suffices to consider the nontrivial case $c_0>0$.   Define
\begin{align}
h(\{c_j\}):=&d\biggl(\prod_j c_j \biggr)^{2/d}- \biggl(\sum_j  c_j\biggr)^2+ (d-1)\sum_j c_j^2,
\end{align}
then it remains to show that $h\geq0$.

First,  consider the special case $c_0=\cdots=c_{d-2}=a$ and $c_{d-1}=b$ with $a>0$ and  $0\leq b\leq a$. Since $h$ is homogeneous, we may assume $a=1$, so that $0\leq b\leq 1$. Then
 \begin{align}
 h(\{c_j\})=&g(b):=db^{2/d}+(d-2)b^2-2(d-1)b.
 \end{align}
 The first and second derivatives of $g(b)$ are given by 
 \begin{equation}\label{eq:gd2}
 \begin{aligned}
 g'(b)=&2b^{(2-d)/d}+2(d-2)b-2(d-1),\\
 g''(b)=&\frac{2(d-2)}{d}\bigl(d-b^{2(1-d)/d}\bigr). 
 \end{aligned}
  \end{equation}
 According to these formulas, it is easy to verify that $g'(b)$ has only two zeros $0<b_0<b_1=1$ in the  interval $0<b\leq 1$. In addition, $g'(b)>0$ when $0< b< b_0$  and $g'(b)<0$ when $b_0<b<1$. Therefore, the minimum of $g(b)$ over the interval  $0\leq b\leq 1$ can only be attained at $b=0$ or $b=1$. Since $g(0)=g(1)=0$, we conclude that $h(\{c_j\})=g(b)\geq 0$ for  $0\leq b\leq 1$, and the inequality is saturated iff $b=0$ or $b=1$.

Next, consider the general case. Since $h$ is homogeneous, we may assume that $\sum_j c_j^2=1$  without loss of generality.  Let $s:=\sum_j  c_j$;
then 
\begin{align}
h(\{c_j\})=d\biggl(\prod_j c_j \biggr)^{2/d}- s^2+ (d-1). 
\end{align}
If $s<\sqrt{d-1}$, then $h>0$. If $s=\sqrt{d-1}$, then $h\geq 0$ and the inequality is saturated iff $c_{d-1}=0$, in which case we have $c_0=c_1=\cdots=c_{d-2}=1/\sqrt{d-1}$. If $s=\sqrt{d}$, then $c_0=c_1=\cdots=c_{d-1}=1/\sqrt{d}$ and $h=0$. 
It remains to consider the scenario $ \sqrt{d-1}<s<\sqrt{d}$, in which case $c_j>0$ for all $j$.

Now, we investigate the minimum of $h(\{c_j\})$  for a given value of $s$. Suppose  the minimum is attained at a given point. 
Using the method of Lagrangian multipliers, it is easy to show that $c_0, c_1, \ldots, c_{d-1}$ take on  two different values, that is, $c_0=\cdots= c_{k-1}> c_k=\cdots =c_{d-1}$, where $1\leq k\leq d-1$. Note that not all $c_j$ can take on the same value due to the constraint $\sum_j c_j^2=1$ and $s<\sqrt{d}$. In the case $d=3$, straightforward calculation  shows that $k=2$, in which case we have
\begin{align}
c_0=c_1=\frac{2s+\sqrt{6-2s^2}}{6}, \quad c_2=\frac{s-\sqrt{6-2s^2}}{3}.
\end{align}
In general, we have $k=d-1$; otherwise,  the values of $c_{k-1},c_k, c_{d-1}$ can be adjusted so  that the value of the product $c_{k-1}c_kc_{d-1}$ decreases, while $c_{k-1}+c_k+ c_{d-1}$ and $c_{k-1}^2+c_k^2+ c_{d-1}^2$ are left invariant, which leads to a contradiction.
The fact  $k=d-1$ implies that 
\begin{equation}
\begin{aligned}
c_0=\cdots=c_{d-2}=&\frac{s+\sqrt{(d-s^2)/(d-1)}}{d},\\
c_{d-1}=&\frac{s-\sqrt{(d-s^2)(d-1)}}{d}. 
\end{aligned}
\end{equation}
Since $c_0=\cdots=c_{d-2}>c_{d-1}$ for $ \sqrt{d-1}<s<\sqrt{d}$, it follows that $h(\{c_j\})>0$ according to the discussion after \eref{eq:gd2}. This observation completes the proof of \lref{lem:ProductBound}.
\end{proof}

\section{Proofs of \psref{pro:GCsym} and \ref{pro:GECsym}}
\begin{proof}[Proof of \pref{pro:GCsym}]
	\Eref{eq:l1RoCsym} can be derived as follows. The equality $C_{l_1}(\rho)=p(d-1)=dF-1$ is easy to verify; the equality $\hat{C}_{l_1}(\rho)= C_{l_1}(\rho)$ follows from the inequality $\hat{C}_{l_1}(\rho)\geq C_{l_1}(\rho)$ and the convexity of $\hat{C}_{l_1}(\rho)$, which implies that $\hat{C}_{l_1}(\rho)\leq p\hat{C}_{l_1}(\psi)= p(d-1)=C_{l_1}(\rho)$; the equality $C_\caR(\rho)= C_{l_1}(\rho)$ follows from Theorem~6 in \rcite{Piani16}. Alternatively, $C_\caR(\rho)$ can be computed based on the symmetry consideration that  $\rho$ is invariant under arbitrary permutations of the basis states. 
	
	To derive \eref{eq:GCsym}, let $\rho_j=|\varphi_j\>\<\varphi_j|$ with 
	\begin{equation}
	|\varphi_j\> =\frac{1}{\sqrt{d-1}}\sum_{k=0, k\neq j}^{d-1}|k\>,\quad j=0, 1,\ldots, d-1. 
	\end{equation}
	Then 
	\begin{equation}
	\frac{1}{d}\sum_{j=0}^{d-1} \rho_j=a|\psi\<\psi|+(1-a)\frac{I}{d}
	\end{equation}
	with $a=(d-2)/(d-1)$. Note that $C_\gc(\rho_j)=0$ for $j=0,1, \ldots, d-1$, we conclude that $C_\gc(\rho)=0$ when $p=(d-2)/(d-1)$, which implies that $C_\gc(\rho)=0$ for  $0\leq p\leq (d-2)/(d-1)$. When $p\geq (d-2)/(d-1)$, we have
	\begin{equation}
	C_\gc(\rho)\leq \frac{p-a}{1-a}C_{\gc}(\psi)=p(d-1)-(d-2)=dF-(d-1). 
	\end{equation}
	On the other hand, the opposite inequality follows from \thref{thm:CgC} and \eref{eq:l1RoCsym} in the main text. This observation confirms \eref{eq:GCsym} and completes the proof of \pref{pro:GCsym}. 
\end{proof}

\begin{proof}[Proof of \pref{pro:GECsym}]
	\Eref{eq:NRoEsym} can be derived as follows. When $1/d^2\leq F\leq 1/d$, we have $\hat{\caN}(\rho)= \caN(\rho)= E_\caR(\rho)=0$ because   $\rho$ is separable.
	When $F\geq 1/d$, the equality $\caN(\rho)=dF-1$ is straightforward  to verify.  The equality $\hat{\caN}(\rho)= \caN(\rho)$ follows from the inequality $\hat{\caN}(\rho)\geq \caN(\rho)$ and the convexity of $\hat{\caN}(\rho)$, which implies that 
	\begin{equation}
	\hat{\caN}(\rho)\leq \frac{F-\frac{1}{d}}{1-\frac{1}{d}}\hat{\caN}(|\Psi\>\<\Psi|)=dF-1=\caN(\rho).
	\end{equation}
	Finally, the equality $E_\caR(\rho)= dF-1$  can be derived  based on the symmetry consideration that  $\rho$ is invariant under the transformation $U\otimes U^*$ for any unitary $U$ (here $U^*$ denotes the complex conjugate of $U$ with respect to a given basis; by contrast, the Hermitian conjugate of $U$  is denoted by $U^\dag$).

	To derive \eref{eq:GECsym}, let 
	\begin{equation}
	|\Phi\> =\frac{1}{\sqrt{d-1}}\sum_{j=0}^{d-2}|jj\>
	\end{equation}
	and  $F_0=|\<\Psi|\Phi\>|^2=(d-1)/d$. Then 
	\begin{align}
	& \int\rmd U \left[(U\otimes U^*)(|\Phi\>\<\Phi|) (U\otimes U^*)^\dag\right]\nonumber\\ &=F_0(|\Psi\>\<\Psi|)+(1-F_0)\frac{I-|\Psi\>\<\Psi|}{d^2-1},
	\end{align}
	where the integral is taken with respect to the normalized Haar measure on the unitary group.
	Observing  that $E_\gc(|\Phi\>\<\Phi|)=0$, we conclude that $E_\gc(\rho)=0$ when $F=F_0$, which further implies that $E_\gc(\rho)=0$ for  $1/d^2\leq F\leq F_0$. When $F\geq F_0$, we have
	\begin{equation}
	E_\gc(\rho)\leq \frac{F-F_0}{1-F_0}E_{\gc}(\Psi)=dF-(d-1). 
	\end{equation}
	On the other hand, the opposite inequality follows from \thref{thm:EgC} and \eref{eq:NRoEsym} in the main text.  This observation confirms \eref{eq:GECsym} and completes the proof of \pref{pro:GECsym}. 
\end{proof}


\begin{thebibliography}{99}
	\bibitem{Horodecki09}R. Horodecki, P. Horodecki, M. Horodecki, K. Horodecki, Quantum entanglement, Rev. Mod. Phys. \textbf{81}, 865 (2009).	
	
	\bibitem{Mandel95} L. Mandel and E. Wolf, \textit{Optical Coherence and Quantum
		Optics} (Cambridge University Press, Cambridge,  1995).
	
	\bibitem{Aberg06} J. Aberg, Quantifying superposition, 	arXiv: quant-ph/0612146.
	
	
	\bibitem{BagaBCH16}E. Bagan, J. A.  Bergou, S. S.  Cottrell,  and M. Hillery, Relations between Coherence and Path Information, Phys. Rev. Lett. \textbf{116}, 160406 (2016). 
	
	
	
	\bibitem{BiswGW17} T. Biswas, M. Garc\'ia D\'iaz, A. Winter, Interferometric visibility and coherence, 	Proc. R. Soc. A \textbf{473}, 20170170 (2017).
	
	
	
	\bibitem{Giovannetti04}V. Giovannetti, S. Lloyd, and L. Maccone, Quantum-Enhanced Measurements: Beating the Standard Quantum Limit, Science \textbf{306}, 1330 (2004).
	
	
	\bibitem{EschMD11}B. M. Escher,	R. L. de Matos Filho, and L. Davidovich, General framework for estimating the ultimate precision limit in noisy quantum-enhanced metrology, Nat. Phys. \textbf{7}, 406-411 (2011).
	
	
	
	\bibitem{MarvS16}I. Marvian, and R. W. Spekkens, How to quantify coherence: Distinguishing speakable and unspeakable notions, Phys. Rev. A \textbf{94}, 052324 (2016).
	
	\bibitem{Shor95} P. W. Shor, Scheme for reducing decoherence in quantum computer memory, Phys. Rev. A \textbf{52}, R2493 (1995).
	
	\bibitem{Hill16}M. Hillery, Coherence as a resource in decision problems: The Deutsch-Jozsa algorithm and a variation, Phys. Rev. A  \textbf{93}, 012111 (2016).
	
	\bibitem{MaYGV16}J. J. Ma, B. Yadin, D. Girolami, V. Vedral, and M. Gu, Converting Coherence to Quantum Correlations, Phys. Rev. Lett. \textbf{116}, 160407 (2016).
	
	\bibitem{Matera16}J. M. Matera, D. Egloff, N. Killoran, M. B. Plenio, Coherent control of quantum systems as a resource theory, Quantum Science and Technology \textbf{1},  01LT01 (2016).
	
	
	\bibitem{HoroO13F}M. Horodecki, J. Oppenheim, Fundamental limitations for quantum and nanoscale thermodynamics, Nat. Commun. \textbf{4}, 2059 (2013).
	
	
	\bibitem{Aberg14} J. Aberg, Catalytic Coherence, Phys. Rev. Lett. \textbf{113}, 150402 (2014).
	
	\bibitem{LostKJR15}M. Lostaglio, K. Korzekwa, D. Jennings, and T. Rudolph, Quantum Coherence, Time-Translation Symmetry, and Thermodynamics, Phys. Rev. X \textbf{5}, 021001 (2015).
	
	
	
	\bibitem{Cwik15} P. \'Cwikli\'nski, M. Studzi\'nski, M. Horodecki,  and J. Oppenheim, Limitations on the Evolution of Quantum Coherences: Towards Fully Quantum Second Laws of Thermodynamics, Phys. Rev. Lett. \textbf{115}, 210403 (2015).
	
	\bibitem{Mitchison15}M. T. Mitchison, M. P. Woods, J. Prior, M. Huber, Coherence-assisted single-shot cooling by quantum absorption refrigerators, New J. Phys. \textbf{17}, 115013 (2015).
	
	\bibitem{BrasB15}J. B. Brask and N. Brunner, Small quantum absorption refrigerator in the transient regime: Time scales, enhanced cooling, and entanglement, Phys. Rev. E \textbf{92}, 062101 (2015).
		
	
	\bibitem{Goold16}J. Goold, M. Huber, A. Riera, L. del Rio, P. Skrzypczyk, The role of quantum information in thermodynamics--a topical review, J. Phys. A: Math. Theor. \textbf{49}, 143001 (2016).
	
	\bibitem{Llobet17}  M. Perarnau-Llobet, E. B\"aumer, K. V. Hovhannisyan, M. Huber, A. Ac\'in, No-Go Theorem for the Characterization of Work Fluctuations in Coherent Quantum Systems, Phys. Rev. Lett. \textbf{118}, 070601 (2017).
	
	
	\bibitem{Engel07}G. S. Engel, T. R. Calhoun, E. L. Read, T. K. Ahn, T. Man\v{c}al, Y. C. Cheng, R. E. Blakenship, and G. R. Fleming, Evidence for wavelike energy transfer through quantum coherence in photosynthetic systems, Nature \textbf{446}, 782 (2007).
	
	\bibitem{Cheng09} Y. C. Cheng and G. R. Fleming, Dynamics of light harvesting in photosynthesis, Annu. Rev. Phys. Chem. \textbf{60}, 241 (2009).
	
	
	
	
	\bibitem{Baumgratz14} T. Baumgratz, M. Cramer, and M. B. Plenio, Quantifying Coherence, Phys. Rev. Lett. \textbf{113}, 140401 (2014).
	
	
	
	\bibitem{HoroO13}M. Horodecki and J. Oppenheim, (Quantumness in the context of) resource theories, Int. J. Mod. Phys. B \textbf{27}, 1345019 (2013).
	
	
	\bibitem{Bran15} F. G. S. L. Brand\~{a}o, G. Gour, Reversible Framework for Quantum Resource Theories, Phys. Rev. Lett. \textbf{115}, 070503 (2015).
	
	
	\bibitem{CoecFS16}B. Coecke, T. Fritz, R. W. Spekkens, A mathematical theory of resources, Information and Computation, \textbf{250},  59-86  (2016).
	
	

	
	\bibitem{LiuHL17}Z.-W. Liu, X. Hu, and S. Lloyd, Resource Destroying Maps, Phys. Rev. Lett. \textbf{118}, 060502 (2017).
	
	\bibitem{WinterYang16} A. Winter, D. Yang, Operational Resource Theory of Coherence,  Phys. Rev. Lett. \textbf{116}, 120404 (2016).
	
	\bibitem{StreAP16}A. Streltsov, G. Adesso, M. B. Plenio, Quantum Coherence as a Resource, arXiv:1609.02439. 

	
	\bibitem{DuBQ15}S. Du, Z. Bai, and X. Qi, Coherence Measures and Optimal Conversion for Coherent States, Quantum Info. Comput. \textbf{15}, 1307 (2015).
	
	\bibitem{YuanZCM15}
	X. Yuan, H. Zhou, Z. Cao, and X. Ma, Intrinsic randomness as a measure of quantum coherence,
	Phys. Rev. A {\bf 92}, 022124 (2015).
	
	
	\bibitem{Chitambar16PRA} E. Chitambar, G. Gour, Comparison of incoherent operations and measures of coherence, Phys. Rev. A \textbf{94}, 052336 (2016).	
	
	\bibitem{Napoli16} C. Napoli, T. R. Bromley, M. Cianciaruso, M. Piani, N. Johnston, and G. Adesso,  Robustness of Coherence: An Operational and Observable Measure of Quantum Coherence, Phys. Rev. Lett. \textbf{116}, 150502 (2016).
		
		
	\bibitem{Piani16}M. Piani, M. Cianciaruso, T. R. Bromley, C. Napoli, N. Johnston,
		and G. Adesso, Robustness of asymmetry and coherence of quantum states, Phys. Rev. A \textbf{93}, 042107  (2016).
		
	\bibitem{QiGY16}X. Qi, T. Gao, F. Yan, Measuring coherence with entanglement concurrence, J. Phys. A: Math. Theor. \textbf{50},  285301 (2017). 
		
	\bibitem{Chin17}S. Chin, Generalized coherence concurrence and path distinguishability, arXiv:1702.06061.	
		
	
	
	\bibitem{Chitambar16PRL} E. Chitambar, G. Gour, Critical Examination of Incoherent Operations and a Physically Consistent Resource Theory of Quantum Coherence, Phys. Rev. Lett. \textbf{117}, 030401 (2016).
	
	
	
	\bibitem{ChitambarH16} E. Chitambar and M. H. Hsieh, Relating the Resource Theories of Entanglement and Quantum Coherence, Phys. Rev. Lett. \textbf{117}, 020402 (2016).
	
	\bibitem{Vicente17}J. I. de Vicente, A. Streltsov, Genuine quantum coherence, J. Phys. A \textbf{50}, 045301 (2017).
	
	

	
	
	\bibitem{ChengH15} S.  Cheng and M. J. W. Hall, Complementarity relations for quantum coherence, Phys. Rev. A \textbf{92}, 042101 (2015).
	
	
	\bibitem{Singh15} U. Singh, M. N. Bera, H. S. Dhar, A. K. Pati, Maximally coherent mixed states: Complementarity between maximal coherence and mixedness, Phys. Rev. A \textbf{91}, 052115 (2015).
	
	\bibitem{Rana16} S. Rana, P. Parashar, and M. Lewenstein, Trace-distance measure of coherence, Phys. Rev. A  \textbf{93}, 012110 (2016).
	
	
	\bibitem{BromCA15}T. R. Bromley, M. Cianciaruso, and G. Adesso, Frozen Quantum Coherence, Phys. Rev. Lett. \textbf{114}, 210401 (2015).
	
	
	
	\bibitem{YadiV16}B. Yadin and V. Vedral, General framework for quantum macroscopicity in terms of coherence, Phys.  Rev. A \textbf{93}, 022122 (2016).
	
	
	\bibitem{YuZXT16}X.-D. Yu, D.-J. Zhang, G. F. Xu, and D. M. Tong, Alternative framework for quantifying coherence, Phys.  Rev. A \textbf{94}, 060302(R) (2016).
	

	
	
	\bibitem{Asboth05}J. K. Asb\'oth, J. Calsamiglia, and H. Ritsch, Computable Measure of Nonclassicality for Light, Phys. Rev. Lett. \textbf{94}, 173602 (2005).
	
	
	\bibitem{StreSDB15}A. Streltsov, U. Singh, H. S. Dhar, M. N. Bera, and G. Adesso, Measuring Quantum Coherence with Entanglement, Phys. Rev. Lett. \textbf{115}, 020403 (2015).
	
	
	\bibitem{VogeS14}W. Vogel and J. Sperling, Unified quantification of nonclassicality and entanglement, Phys. Rev. A \textbf{89}, 052302 (2014).
	
	
	\bibitem{DuBG15}S. Du, Z. Bai, and Y. Guo, Conditions for coherence transformations under incoherent operations, Phys. Rev. A \textbf{91}, 052120 (2015); Erratum, Phys. Rev. A \textbf{95}, 029901(E) (2017).
	
	\bibitem{QiBD15}X. Qi, Z. Bai, and S. Du, Coherence convertibility for mixed states, arXiv:1505.07387.
	
	\bibitem{XiLF15} Z. J. Xi, Y. M. Li,  and H. Fan, Quantum coherence and correlations in quantum system, Sci. Rep. \textbf{5}, 10922 (2015).
	
	
	\bibitem{YaoXGS15}Y. Yao, X. Xiao, L. Ge, and C. P. Sun, Quantum coherence in multipartite systems, Phys. Rev. A  \textbf{92}, 022112 (2015).
	
	
	\bibitem{ChitambarSRB16} E. Chitambar, A. Streltsov, S. Rana, M. N. Bera, G. Adesso, and M. Lewenstein, Assisted Distillation of Quantum Coherence, Phys. Rev. Lett. \textbf{116}, 070402 (2016).
	
	\bibitem{StreltsovCRB16} A. Streltsov, E. Chitambar, S. Rana, M. N. Bera, A. Winter, and M. Lewenstein, Entanglement and Coherence in Quantum State Merging, Phys. Rev. Lett. \textbf{116}, 240405 (2016).
	
	\bibitem{Killoran16} N. Killoran, F. E. S. Steinhoff, and M. B. Plenio, Converting Nonclassicality into Entanglement, Phys. Rev. Lett. \textbf{116}, 080402 (2016).
	
	\bibitem{ChenGJL16}
	J. Chen,  S. Grogan, N. Johnston, C.-K. Li, and S. Plosker, Quantifying the coherence of pure quantum states, Phys. Rev. A \textbf{94}, 042313 (2016).
	
	\bibitem{AdesBC16}G. Adesso, T. R. Bromley, and M. Cianciaruso,
	Measures and applications of quantum correlations, J. Phys. A: Math. Theor. \textbf{49}, 473001 (2016).
	
	\bibitem{Gour05}G. Gour, Family of concurrence monotones and its applications, Phys. Rev. A \textbf{71}, 012318 (2005).
	
	
	\bibitem{Vida00} G.  Vidal, Entanglement monotones, J.  Modern Opt. \textbf{47}, 355 (2000).
	
	
	
	
	\bibitem{VollW01}K. G. H. Vollbrecht and R. F. Werner, Entanglement measures under symmetry,
	Phys. Rev. A \textbf{64}, 062307 (2001).
	
	
	
	\bibitem{ChenAF05}K. Chen, S. Albeverio, and S.-M. Fei, Concurrence of Arbitrary Dimensional Bipartite Quantum States, Phys. Rev. Lett. \textbf{95}, 040504 (2005). 
	
	\bibitem{TothMG15}G. T\'oth, T.  Moroder, and O. G\"uhne, Evaluating Convex Roof Entanglement Measures, Phys. Rev. Lett.  \textbf{114}, 160501 (2015).
	
	\bibitem{SentEGH16}G. Sent\'is, C. Eltschka, O. G\"uhne, M. Huber, and J. Siewert,
	Quantifying Entanglement of Maximal Dimension in Bipartite Mixed States, Phys. Rev. Lett.  \textbf{117}, 190502 (2016). 
	
	
	\bibitem{GiraG17}M. W. Girard and G. Gour,
	Entanglement monotones and transformations of symmetric bipartite states, Phys. Rev. A \textbf{95}, 012308 (2017).
	
	
	\bibitem{MarsOA11book}A. W. Marshall and I. Olkin and B. C. Arnold, \textit{Inequalities: Theory of Majorization and Its Applications}, Springer Series in Statistics, second ed. (Springer 2011).
	
	\bibitem{Bhat97book}R. Bhatia, \textit{Matrix Analysis} (Springer, New York, 1997).
	
	\bibitem{Niel99}  M. A. Nielsen, Conditions for a Class of Entanglement Transformations, Phys. Rev. Lett. \textbf{83}, 436 (1999).
	
	
	\bibitem{Vida99}G. Vidal, Entanglement of Pure States for a Single Copy, Phys. Rev. Lett. \textbf{83}, 1046 (1999).
	
	\bibitem{ZyczB02}K. \.{Z}yczkowski, I. Bengtsson, 
	Relativity of Pure States Entanglement,
	Ann. Phys. \textbf{295}, 115-135 (2002).
	
	\bibitem{Turg07}S. Turgut, Catalytic transformations for bipartite pure states, 
	J. Phys. A: Math. Theor. \textbf{40}, 12185-12212 (2007).
	
	
	\bibitem{BuSW16}K. Bu, U. Singh, and J. Wu, 
	Catalytic coherence transformations,
	Phys. Rev. A \textbf{93}, 042326 (2016).
	
	
	\bibitem{Rain99}E. M. Rains, Bound on distillable entanglement, Phys. Rev. A \textbf{60}, 179 (1999).
	
	
	\bibitem{ZhuHC17A}H. Zhu, M. Hayashi, L. Chen, Axiomatic and operational connections between $l_1$-norm of coherence and negativity, arXiv:1704.02896.
	
	\bibitem{ZhuHC17C}
	H. Zhu, M. Hayashi, and L. Chen,
	Coherence and entanglement measures based on R\'{e}nyi relative entropies,
	arXiv:1706.00390.
	
	\bibitem{MartCP16}O. Marty, M. Cramer, and M. B. Plenio, Practical Entanglement Estimation for Spin-System Quantum Simulators, Phys. Rev. Lett. \textbf{116}, 105301 (2016). 
	
	\bibitem{MoroBLH13}T. Moroder, J.-D. Bancal, Y.-C. Liang, M. Hofmann, and O. G\"uhne, Device-Independent Entanglement Quantification and Related Applications, Phys. Rev. Lett.  \textbf{111}, 030501 (2013). 
	
	
	
	
\end{thebibliography}
\end{document}